\documentclass[aps,showpacs,pra,twocolumn,superscriptaddress]{revtex4-2}

\usepackage{exscale,amsfonts,mathrsfs}
\usepackage{bbm}
\usepackage{graphicx}
\usepackage{latexsym}
\usepackage{times}
\usepackage[T1]{fontenc}
\usepackage{enumerate}
\usepackage{bbold}
\usepackage{color}
\usepackage{mathtools}
\usepackage[normalem]{ulem}
\usepackage{amsfonts,amsmath,amssymb,amsthm,dsfont}
\usepackage[colorlinks=true,citecolor=blue,urlcolor=blue]{hyperref}
\usepackage{natbib}
\usepackage{physics}
\usepackage{changes}

\newcommand{\expect}[1]{\left\langle#1\right\rangle}

\theoremstyle{plain}
\newtheorem{thm}{Theorem}
\newtheorem{lem}{Lemma}
\newtheorem{fact}{Fact}
\newtheorem{defin}{Definition}

\newtheorem{cor}{Corollary}

\begin{abstract}

The incompatibility of measurements is the key feature of quantum theory that distinguishes it from the classical description of nature. Here, we consider groups of $d$-outcome quantum observables with prime $d$ represented by non-Hermitian unitary operators whose eigenvalues are $d$th roots of unity. We additionally assume that these observables mutually commute up to a scalar factor being one of the $d$'th roots of unity. By representing commutation relations of these observables via a frustration graph, we show that for such a group, there exists a single unitary transforming them into a tensor product of generalized Pauli matrices and some ancillary mutually commuting operators. Building on this result, we derive upper bounds on the sum of the squares of the absolute values and the sum of the expected values of the observables forming a group. Such bounds are of particular importance to deriving uncertainty relations or constructing entanglement witnesses, and are also useful in inflation technique. We finally utilize these bounds to compute the generalized geometric measure of entanglement for qudit stabilizer subspaces.
\end{abstract}

\begin{document}

\title{Frustration graph formalism for qudit observables}
\author{Owidiusz Makuta}
\affiliation{Instituut-Lorentz, Universiteit Leiden, P.O. Box 9506, 2300 RA Leiden, The Netherlands}
\affiliation{$\langle \text{aQa}^\text{L} \rangle$ Applied Quantum Algorithms Leiden, The Netherlands}
\affiliation{Center for Theoretical Physics, Polish Academy of Sciences, Aleja Lotnik\'{o}w 32/46, 02-668 Warsaw, Poland}

\author{Błażej Kuzaka}
\affiliation{Center for Theoretical Physics, Polish Academy of Sciences, Aleja Lotnik\'{o}w 32/46, 02-668 Warsaw, Poland}

\author{Remigiusz Augusiak}
\email{augusiak@cft.edu.pl}
\affiliation{Center for Theoretical Physics, Polish Academy of Sciences, Aleja Lotnik\'{o}w 32/46, 02-668 Warsaw, Poland}

\maketitle

\section{Introduction}

The incompatibility of measurements is one of the most fundamental features that differentiates quantum mechanics from the classical description of nature. It has been found to be a necessary component for the existence of Bell nonlocality \cite{PhysRevLett.113.160402}, so incompatibility is also instrumental to all applications that exploit Bell nonlocality such as quantum key distribution \cite{743501, Zapatero2023-kf, Nadlinger_2022}, self-testing \cite{_upi__2020, PhysRevLett.125.260507} or randomness certification \cite{Acin2016-xj, Pironio2010-ep}. Moreover, it plays an important role in quantum metrology, as it is directly related to the limits of measurement accuracy described by uncertainty relations \cite{PhysRevLett.113.260401, PhysRevA.98.042121, PhysRevA.100.032118}. 

In the study of this phenomenon, one can limit oneself to projective measurements represented by operators called observables. Within this representation, the measurement incompatibility translates to non-commutativity of the corresponding observables. Interestingly, this limitation allows us to apply the same technique used to study non-commutativity to a broader class of problems, many of which are not directly related to measurement incompatibility. An example of such a problem is that of determining a ground-state energy of a given Hamiltonian \cite{diep2024frustratedspinsystemshistory}, which is often a highly nontrivial task, partly due to the potential noncommutativity between different parts of the Hamiltonian. Therefore, understanding the exact consequences of noncommutativity between different operators is vital to this problem and to a broader class of problems in quantum information.

Recently, this topic has been studied with regard to a set of dichotomic observables \cite{10.1145/3519935.3519960, Chapman2020characterizationof, chapman2023unifiedgraphtheoreticframeworkfreefermion, PhysRevA.107.062211, Xu_2024, aguilar2024classificationpauliliealgebras, Mor_n_2024}. Of particular interest for this work is the result of Ref. \cite{PhysRevA.107.062211}, in which the authors derived an upper bound on a sum of squares of expected values of such observables. Such quantities have a plethora of applications: from the derivation of uncertainty relations and entanglement witnesses, to their frequent use in the inflation technique for quantum networks (see e.g. Ref. \cite{Hansenne_2022}). Importantly, the results of Ref. \cite{PhysRevA.107.062211} have been derived by representing the anticommutation relations of the elements of a given set of observables using the so-called anticommutation graph. This representation makes the anticommutation structure of the set explicit, which is key to deriving the upper bound on sum of squares. This bound
turns out to be the Lovász number \cite{PhysRevA.107.062211} of the graph, confirming the essential role of graph formalism in deriving it. However, this upper bound, is not saturable for all sets of observables, which is an issue for many possible applications of these types of results. While it was speculated that another graph property called clique number could constitute a tight bound on the sum of squares of expectations, this was ultimately shown not to be the case by Z.-P. Xu and collaborators in Ref. \cite{Xu_2024}. 

In this work, our aim is to develop this line of research in two ways. First, we show that if the set of observables forms a group, then the clique number is in fact a tight upper bound on the sum of squares of their expected values. Second, similarly to the definition put forward in Ref. \cite{mann2024graphtheoreticframeworkfreeparafermionsolvability}, we extend the formalism of the anticommutation graph by considering generalized $d$-outcome observables that commute up to a scalar factor of the $d$'th root of unity for some prime $d$ (see Ref. \cite{Sarkar_2024} for a related study of such observables). 

This not only extends the study of the properties of frustration graphs but also paves the way for new applications. Most importantly, it enables the generalization of the methods based on the sum-of-squares bounds as uncertainty relations and entanglement witnesses to a much broader class of states and operators, including those with local Hilbert space dimensions greater than two.

The core idea behind these results is to transform the observables into a tensor product of generalized Pauli matrices and some ancillary operators that are pair-wise commuting, which helps to significantly simplify the derivations. To this end, we generalize the self-testing statement from Ref. \cite{PhysRevA.106.012431} to a multi-operator case. This naturally leads us to use the stabilizer formalism, a framework originally developed to construct quantum error correction codes \cite{PhysRevLett.77.793, PhysRevA.103.042420, gottesman1997stabilizercodesquantumerror} that later found widespread use in the study of entanglement and non-locality \cite{PhysRevA.72.022340, Makuta_2021}. Using this formalism, we derive upper bounds on the sum of the squares of the absolute values and the sum of the expected values of the observables forming a group.

We then utilize the "sum of squares bound" to analytically compute the geometric measure and the generalized geometric measure of entanglement \cite{PhysRevA.81.012308} for qudit stabilizer subspaces. Surprisingly, we found that for a given prime local dimension, the generalized geometric measure of entanglement of a genuine multipartite entangled stabilizer subspace can only take one value: $(d-1)/d$.

The rest of the manuscript is structured as follows: in \textit{Preliminaries} we introduce concepts from graph theory, group theory, and entanglement theory necessary for the understanding of the results of this work. Next, in \textit{Unitary equivalence via frustration graph} we show how to utilize the self-testing toolbox in order to largely simplify our problem. In \textit{Applications}, we present the possible use-case of our technique in deriving bounds on sums, and sums of squares, of expected values of $d$-outcome observables, as well as the derivation of the generalized geometric measure of entanglement for qudit stabilizer subspaces. We conclude the manuscript by discussing open questions and potential research directions.

\section{Results}
\subsection*{Preliminaries}

\textit{(1) Graphs}. 
We begin from the introduction of the most instrumental tool in this work: graph theory. A \textit{graph} $G$ is defined as an ordered pair $G=(V, E)$ consisting of the set of vertices $V$ and the set of edges $E$. In this work, we consider two types of graphs: simple graphs and weighted, directed graphs. The former are graphs for which edges are undirected and there are no edges connecting a vertex to itself, and the latter are graphs composed of directed edges such that each edge has an assigned weight. For simplicity, we will often refer to both of these graphs as just graphs, making a distinction only when it is relevant.

A subgraph $G_{S}=(V_{S},E_{S})$ of a graph $G$ is a graph for which $V_{S}\subseteq V$ and $E_{S}\subseteq E$. A special case of a subgraph is a \textit{clique} $C$ in which every ordered pair of distinct vertices is connected by an edge. For a given graph $G$ we define a \textit{clique number} $\tilde{\omega}(G)$ as the number of vertices in the largest clique of $G$. 

Another notion from graph theory which is relevant to this work is that of proper $l$-coloring of a simple graph $G=(V, E)$ which is a partition of the set $V=V_{1}\cup \ldots \cup V_{l}$ into disjoint sets $V_{k}$, such that if $i \in V_{k}$ then $j\in V_{k}$ only if $(i,j)\notin E$. The standard interpretation of the proper $l$-coloring is that vertices in a graph $G$ are colored with $l$ different colors, such that two vertices connected by an edge have to be colored with a different color. Then, the smallest $l$ for which a proper $l$-coloring of a given $G$ exists is called a \textit{chromatic number} $\chi(G)$ of a graph $G$.

Lastly, in order to efficiently perform operation on a given graph, one can make use of a square matrix representation of a graph in terms of the \textit{adjacency matrix $\Gamma$} whose entry $\Gamma_{ij}\in \{0,\ldots,d-1\}$ encodes the number of edges connecting a pair of vertices $i,j\in V$.

\textit{(2) Groups.} A group $\mathcal{G}$ is a non-empty set equipped with a binary and associative operation $\odot$ such that there exists a neutral element and each element has its inverse. We call a subset $\{t_{i}\}_{i}\subset\mathcal{G}$ a \textit{generating set} of the group if any element $a\in\mathcal{G}$ can be expressed via a combination of finitely many elements $t_{i}$ called generators. We often denote this fact by writing $\mathcal{G}=\langle t_{1},t_{2},\ldots \rangle_{\odot}$. Lastly, a \textit{subgroup} $\mathcal{G}'$ of a group $\mathcal{G}$ is a subset $\mathcal{G}'\subset\mathcal{G}$ closed under the operation $\odot$.

\textit{(3) Genuine multipartite entanglement}. Let us consider a scenario where $N$ parties share a pure state $\ket{\psi}\in \mathcal{H}=\bigotimes_{i=1}^{N} \mathcal{H}_{i}$, where $\mathcal{H}_{i}$ is a Hilbert subspace associated with the $i$'th party. We call $\ket{\psi}$ \textit{genuinely multipartite entangled} (GME) iff it cannot be represented as a tensor product of two other vectors across any bipartition of the set $[N]:=\{1,\ldots,N\}$ into two non-empty and disjoint sets $Q,\overline{Q}\subset [N]$; in what follows we denote such a bipartition $Q|\overline{Q}$. In other words, $\ket{\psi}\neq \ket{\psi_Q}\otimes|\psi_{\overline{Q}}\rangle$ for two states $\ket{\psi_Q}$ and $|\psi_{\overline{Q}}\rangle$ and any bipartition $Q|\overline{Q}$.

In the mixed-state case, the definition of GME is slightly more involved. A mixed state $\rho\in \mathcal{B}(\mathcal{H})$ is GME \cite{PhysRevA.65.012107} if it cannot be decomposed as a probabilistic mixture of states that are separable across different bipartitions $Q|\overline{Q}$. Formally, we say a state $\rho$ is GME if
\begin{equation}
\rho\neq \sum_{Q\subset [N]}q_{Q}\sum_{\lambda} p_{\lambda,Q} \rho^{(\lambda)}_{Q}\otimes \rho ^{(\lambda)}_{\overline{Q}}
\end{equation} 
for any $\rho^{(\lambda)}_{Q}\in \mathcal{B}(\bigotimes_{i\in Q} \mathcal{H}_{i})$, $\rho^{(\lambda)}_{\overline{Q}}\in \mathcal{B}(\bigotimes_{i\in \overline{Q}} \mathcal{H}_{i})$ and any probability distributions $\{q_{Q}\}$, $\{p_{\lambda,Q}\}$.

{Let us finally mention that the definition of genuine entanglement also extends to subspaces of $\mathcal{H}$: a subspace $W\subset \mathcal{H}$ is GME if all the pure states from $W$ are genuinely multipartite entangled. In other words, a subspace $W$ is GME iff it is void of pure product state. Importantly, if $W$ is GME, then every density matrix defined on it is GME, too. Thus, investigation of the entanglement properties of subspaces of multipartite Hilbert spaces provides a new approach towards the characterization of multipartite entanglement. 

\textit{(4) Generalised geometric measure of entanglement}. One of the most popular quantifiers of entanglement of pure states is the \textit{geometric measure of entanglement} \cite{Shimony_95,HBarnum_2001}. For a given state $\ket{\phi}$, and for a given bipartition $Q|\overline{Q}$, it}is defined through the following formula
\begin{equation}\label{eq measure ent states}
E_{\textrm{GM}}^{Q}(\ket{\phi})=1-\max_{\ket{\psi}\in\Phi_{Q}} 
\left|\bra{\psi}\ket{\phi}\right|^{2},
\end{equation}
where $\Phi_{Q}$ denotes the set of all pure states that are product across $Q|\overline{Q}$. Here, two notes are in order. First, this measure can also be defined for mixed states; however, since our results do not require using the more general definition, we omit it here. Second, it is often the case that for a multipartite system, the maximization is over the set of fully product states. Please note that this is not the case here; in this work, we consider the bipartite geometric measure of entanglement, in which the parties, which can be associated with sets $Q$ and $\overline{Q}$, are made up of many qudits.

In order to quantify the amount of genuine entanglement of a state $\ket{\phi}$ one uses the so-called \textit{generalized geometric measure of entanglement} (GGM) \cite{PhysRevA.81.012308}, which is defined as the minimum of $E_{\textrm{GM}}^{Q}$ over all bipartitions $Q|\overline{Q}$,
\begin{equation}
E_{\textrm{GGM}}(\ket{\phi})= \min_{Q|\overline{Q}} E_{\textrm{GM}}^{Q}(\ket{\phi}).
\end{equation}
It is worth noticing that $E_{\mathrm{GGM}}(\ket{\phi})$ is nonzero iff $\ket{\phi}$ is genuinely entangled.

Interestingly, the above entanglement measures can be generalized to quantify the amount of entanglement present in subspaces. In fact, following \cite{PhysRevA.76.042309}, one defines for a given subspace $\mathcal{V}\subset\mathcal{H}$ the following quantities

\begin{equation}\label{eq measure ent subspaces_1}
   E_{\textrm{GM}}^{Q}(\mathcal{V})=\min_{\ket{\phi}\in \mathcal{V}}E_{\textrm{GM}}^{Q}(\ket{\phi}),
\end{equation}
and
\begin{equation}\label{eq measure ent subspaces_2}
  E_{\textrm{GGM}}(\mathcal{V})=\min_{\ket{\phi}\in \mathcal{V}}E_{\textrm{GGM}}(\ket{\phi}),
\end{equation}
which quantify, respectively, the minimal geometric measure of entanglement across a fixed bipartition and the minimal generalized geometric measure of entanglement of all vectors in $\mathcal{V}$.

It is crucial to mention that the above measures can also be expressed in terms of the projection $\mathcal{P}_\mathcal{V}$ onto the subspace $\mathcal{V}$ as \cite{PhysRevA.82.012327, Demianowicz_2019} 
\begin{align}\label{eq:GGM final form}
    \begin{split}
     E_{\textrm{GM}}^{Q}(\mathcal{V})&=1-\max_{\ket{\psi}\in\Phi_{Q}}\bra{\psi}\mathcal{P}_{\mathcal{V}}\ket{\psi},\\
    E_{\textrm{GGM}}(\mathcal{V})&=1-\max_{Q|\overline{Q}}\max_{\ket{\psi}\in\Phi_{Q}}\bra{\psi}\mathcal{P}_{\mathcal{V}}\ket{\psi},
    \end{split}
\end{align}
where $\mathcal{P}_{\mathcal{V}}$ is the projector onto the subspace $\mathcal{V}$. 

\textit{(5) Stabilizer formalism}. Let us now assume that $\mathcal{H}_i=\mathbb{C}^d$ for all $i=1,\ldots,N$ and the generalized Pauli matrices $X$ and $Z$ acting on $\mathbb{C}^d$ defined as
\begin{equation}
X\coloneqq \sum_{j=0}^{d-1} \ket{j+1}\! \bra{j},\qquad Z\coloneqq \sum_{j=0}^{d-1}\omega^{j}\ket{j}\!\bra{j},
\end{equation}
where $\ket{d}\equiv \ket{0}$ and $\omega = \exp(2\pi \mathbb{i}/d)$ and $\mathbb{i}$ is the imaginary unit. Let then $W_{\mathbf{i},\mathbf{j}}$ be an operator acting on an $N$-qudit Hilbert space $(\mathbb{C}^{d})^{\otimes N}$ given by
\begin{equation}
W_{\mathbf{i},\mathbf{j}} \coloneqq \mu_{\mathbf{i},\mathbf{j}} \bigotimes_{j=1}^{N}X^{i_{j}}Z^{{i_{j}}},
\end{equation}
where $\mathbf{i}=\{i_{1},\ldots,i_{N}\}$ and $\mathbf{j}=\{j_{1},\ldots,j_{N}\}$ are binary strings, and $\mu_{\mathbf{i},\mathbf{j}}\in\{1,\mathbb{i}\}$ is chosen so that $W_{\mathbf{i},\mathbf{j}}^{d}=\mathbb{1}$ with $\mathbb{1}$ being the identity operator. We define a set $\tilde{\mathbb{P}}_{N}$ to be a set of all $W_{\mathbf{i},\mathbf{j}}$ for a given $N$.

Then, a Pauli group $\mathbb{P}_{N}$ is defined as the set
\begin{equation}
 \{\omega^{j}M\; | \; M\in \tilde{\mathbb{P}}_{N},\: j\in R_{d}\},
\end{equation}
with matrix multiplication as a group operation, where $R_{d}=\mathbb{Z}_{d}$ for odd $d$ and $R_{d}=\{0,1/2,1,\ldots,d-1/2\}$ for even $d$.

A \textit{stabilizer} $\mathbb{S}$ is an abelian subgroup of the Pauli group $\mathbb{S}\subset\mathbb{P}_{N}$ with an additional constraint that $\varphi \mathbb{1}\in\mathbb{S}$ only if $\varphi = 1$. For simplicity, it is convenient to describe the stabilizer $\mathbb{S}$ via its generating set, which we denote by $\{g_{i}\}_{i=1}^{k}$.

The most important feature of a stabilizer is that it defines a subspace $\mathcal{V}_{\mathbb{S}}\subset \mathcal{H}$. First, we say that a state $\ket{\psi}$ is stabilized by $\mathbb{S}$ if for all $s\in \mathbb{S}$
\begin{equation}
s\ket{\psi}=\ket{\psi}.
\end{equation}
Then, for any stabilizer $\mathbb{S}$ we can find a corresponding \textit{stabilizer subspace} $\mathcal{V}_{\mathbb{S}} \subset (\mathbbm{C}^d)^{\otimes N}$, which is a space containing all states $\ket{\psi}$ stabilized by $\mathbb{S}$.

Recall that the stabilizer $\mathbb{S}$ consists of $N$-fold tensor products of $W_{i,j}$ matrices. From this construction it follows that every element $s\in\mathbb{S}$ can be decomposed with respect to every bipartition $Q|\overline{Q}$ in the following manner
\begin{equation}
    s=s^{(Q)}\otimes s^{(\overline{Q})},
\end{equation}
where $s^{(Q)}$ and $s^{(\overline{Q})}$ act on the Hilbert spaces associated to the parties from $Q$ and from $\overline{Q}$ respectively. Due to the fact that all operators $s$ mutually commute, and the fact that $ZX=\omega XZ$, for any bipartition $Q|\overline{Q}$ and any $s_{i},s_{j}\in\mathbb{S}$, we have
\begin{equation}\label{eq:reduced_stab_commutatiton}
\left[s_i^{(Q)},s_j^{(Q)}\right]_{\bullet}=\omega^{\tau_{i,j;Q}} \mathbb{1}\quad\mathrm{and}\quad   \left[s_i^{(\overline{Q})},s_j^{(\overline{Q})}\right]_{\bullet}=\omega^{-\tau_{i,j;Q}}\mathbb{1},
\end{equation}
where $[A,B]_{\bullet}=ABA^{-1}B^{-1}$ and $\tau_{i,j;Q}\in \mathbb{Z}_{d}$. Notice, that the same $\tau_{i,j}$ appears in both commutation relations, but with opposite sign. This is due to the fact that $s_{i}$ and $s_{j}$ commute, which implies
\begin{equation}
\left[s_{i}^{(Q)},s_{j}^{(Q)}\right]_{\bullet}\otimes\left[s_{i}^{(\overline{Q})},s_{j}^{(\overline{Q})}\right]_{\bullet}=\mathbb{1}.
\end{equation}

This notation is particularly useful in the formulation of a necessary and sufficient condition for genuine multipartite entanglement of stabilizer subspaces (see Ref. \cite{Makuta2023fullynonpositive}), which we here state as the following fact.
\begin{fact}\label{fakt1}
Consider a stabilizer $\mathbb{S}=\langle g_{1},\ldots,g_{k}\rangle$. 
For every bipartition $Q|\overline{Q}$ there exist a pair $i,j\in [k]$ such that 
\begin{equation}\label{fact1cond}
\left[g_i^{(Q)},g_j^{(Q)}\right]_{\bullet}\neq \mathbb{1},
\end{equation}
iff the stabilizer subspace $V_{\mathbb{S}}$ is genuinely multipartite entangled.
\end{fact}

\subsection*{Unitary equivalence via frustration graph}

In this section, we derive a certain form of a self-testing statement for any collection of unitary operators whose eigenvalues are powers of $\omega=\exp(2\pi\mathbbm{i}/d)$ for some prime $d$, and which obey certain commutation relations. In fact, we show that for any such a collection, there exists another unitary operation that transforms all the operators into a tensor product of the generalized Pauli operators $X^{i}Z^{j}$ and some ancillary pair-wise commuting operators.

To start, let us consider a set of operators $\{T_{i}\}_{i=1}^{k}$ acting on some arbitrary finite-dimensional Hilbert space $\mathcal{H}=\mathbb{C}^d$ such that each $T_{i}$ is unitary, $T_{i}^{d}=\mathbb{1}$, and for every pair $T_{i},T_{j}$ we have $[T_{i},T_{j}]_{\bullet}=\omega^{l} \mathbb{1}$ for some $l\in \mathbb{Z}_{d}$. Next, let $I \in \mathbb{Z}_{d}^{k}$ be a $k$ element vector with entries in $\mathbb{Z}_{d}$. For each such $I$ we define an operator $A_{I}$ as
\begin{equation}\label{eq:def_A_I}
A_{I}= \alpha_{I}\prod_{i=1}^{k} T_{i}^{I_{i}},
\end{equation}
where $\alpha_{I}\in\{1,\mathbb{i}\}$ is chosen to satisfy the condition $A_{I}^{d}=\mathbb{1}$ (see Appendix \ref{app:A_I} for a short proof of this fact.). 

\begin{defin}\label{def: A}
A set $\mathcal{A}=\{A_{I}\}_{I \in \mathbb{Z}_{d}^{k}}$ is a group of unitary matrices $A_{I}$ defined in Eq. \eqref{eq:def_A_I} such that
\begin{itemize}
    \item $A_{I}^{d}=\mathbb{1}$ for any $I$,
    \item for any pair $A_{I},A_{J}\in\mathcal{A}$ we have $[A_{I},A_{J}]_{\bullet}=\omega^{l}\mathbb{1}$ for some $l\in \mathbb{Z}_{d}$.
\end{itemize}
\end{defin}

Notice that $\{T_{i}\}_{i=1}^{k}$ is a generating set of $\mathcal{A}$ and the group operation in $\mathcal{A}$ is given by $A_{I}\odot A_{J}=A_{I+ J}$, so that $I+J\in\mathbb{Z}_{d}^{k}$. However, it is important to state that in this work, we are interested in commutation relations of elements $\mathcal{A}$ with respect to regular matrix multiplication, not group operation $\odot$ since $\mathcal{A}$ is an abelian group with respect to $\odot$. 

For completeness, let us also mention that the set $\mathcal{A}$ can be understood as a representation of a generalized quasi-Clifford algebra, i.e. an algebra of elements which mutually commute up to roots of unity of degree $d$, and whose $d$'th powers equal the identity. However, no knowledge of Clifford algebras is required to follow the remainder of the paper.

Let us also define two substructures of $\mathcal{A}$ that will be of particular interest in this work. First, we define $\mathcal{S}(\mathcal{A})\subset\mathcal{A}$ to be the largest subgroup of $\mathcal{A}$ for which there exists exactly one element $\mathcal{A}\in \mathcal{S}(\mathcal{A})$ that commutes with every element from $\mathcal{A}$. Second, $\mathcal{C}(\mathcal{A})\subset\mathcal{A}$ is defined to be the largest subgroup of $\mathcal{A}$ such that for all $A_{C}\in \mathcal{C}(\mathcal{A})$ and all $A\in \mathcal{A}$ we have
\begin{equation}
[A_{C},A]_{\bullet} =\mathbb{1}.
\end{equation}
Let us note here that despite its apparent similarities, $\mathcal{C}(\mathcal{A})$ is not a center of $\mathcal{A}$. The difference again boils down to matrix multiplication not being the group operation of $\mathcal{A}$ - under operation $\odot$ all elements of $\mathcal{A}$ mutually commute (so $\mathcal{A}$ is its own center), which however is not the case under matrix multiplication.

\subsubsection*{Frustration graphs}
As we aim to study the structure arising from the commutation relations of elements of $\mathcal{A}$, we need a formalism that neatly encodes them. A \textit{frustration graph} $G=(V,E)$ is a weighted, directed graph for which each vertex corresponds to an element from $\mathcal{A}$ and the weights of the edges $\Gamma_{I,J}$ are given by 
\begin{equation}\label{eq:Gamma_def}
[A_{I},A_{J}]_{\bullet}=\omega^{\Gamma_{I,J}}\mathbb{1},
\end{equation}

Note that the notion of a frustration graph to represent such commutation relations was already proposed in \cite{mann2024graphtheoreticframeworkfreeparafermionsolvability} in the context of qudit Hamiltonians. However, our definition slightly differs from that one because: (i) in Ref. \cite{mann2024graphtheoreticframeworkfreeparafermionsolvability}, the commutation relations are limited to $[A,B]_{\bullet} = \omega^{\kappa}\mathbb{1}$ for $\kappa\in\{-1,0,1\}$, whereas in our case we consider $\kappa\in \{0,\ldots, d-1\}$; (ii) the relation $[A,B]_{\bullet} = \omega\mathbb{1}$ is represented in the graph in\cite{mann2024graphtheoreticframeworkfreeparafermionsolvability} as $\Gamma_{A,B}=0$ and $\Gamma_{B,A}=1$ whereas in our definition, the same relation is encoded as $\Gamma_{A,B}=d-1$ and $\Gamma_{B,A}=1$.

Since the group $\mathcal{A}$ can be fully described by its generators $\mathcal{A}=\langle T_{1},\dots,T_{k}\rangle_{\odot}$ a natural question arises about the frustration subgraphs describing only relations between the generators. As it turns out, these graphs play an instrumental role in this work. We define a generating graph $g$ to be a graph for which each vertex corresponds to a generator of $\mathcal{A}$ with the corresponding adjacency matrix $\gamma$ defined by
\begin{equation}\label{eq:gamma_def}
[T_{i},T_{j}]_{\bullet}=\omega^{\gamma_{i,j}}\mathbb{1}.
\end{equation}
Notice that in stark contrast to the graph $G$, $g$ is not unique for a given $\mathcal{A}$ as it can be generated by many different generating sets. 

This non-uniqunes of $g$ translates neatly into the transformation of the corresponding adjacency matrix $\gamma$: for any two generating graphs $g$ and $g'$ corresponding to the same frustration graph $G$ there exist an invertible operator $O\in M_{k\cross k}(\mathbb{Z}_{d})$. In fact, there exists a one-to-one correspondence between the generating set transformation and generating graph transformation (see Lemma \ref{lem:O_transformation} in Appendix \ref{app:theorem_1} for the proof.)

True to its name, the generating graph $g$ can be used to generate the frustration graph via 
\begin{equation}\label{eq:Gamma_gamma}
\Gamma_{I,J} = \sum_{i=1}^{k}\sum_{j=1}^{k}I_{i}J_{j}\gamma_{i,j} \mod{d},
\end{equation}
which follows from Eqs. \eqref{eq:def_A_I}, \eqref{eq:Gamma_def}, and $\eqref{eq:gamma_def}$. 

We can also consider a related graph based on the commutation relations of the elements of $\mathcal{A}$ called \textit{commutation graph}, which we denote as $\overline{G}$. In stark contrast to $G$, $\overline{G}$ is neither weighted nor directed; $\overline{G}$ is a simple graph in which the vertices corresponding to $A_{I}$ and $A_{J}$ are connected iff $A_{I}$ and $A_{J}$ commute and $I\neq J$. Notice that for $d=2$, $\overline{G}$ is a complement of $G$, i.e., two vertices in $G$ are connected iff they are not connected in $\overline{G}$.

\subsubsection*{Unitary transformation of $\mathcal{A}$}\label{subsec:self-testing}

The last tool we need to formulate the main result of this section is a simple generalization of a self-testing related result from Ref. \cite[Lemma 6]{PhysRevA.106.012431}, stated as the following lemma.
\begin{lem}\label{lem:self-testing}
Let us consider a set of unitary operators $\{M_{1},M_{2},\dots,M_{2m}\}$ acting on some finite-dimensional Hilbert space $\mathcal{H}$ such that $M_i^d=\mathbbm{1}$ and for every pair $i\neq j$, $M_iM_j=\omega^lM_jM_i$ for some $l\in\{0,\ldots,d-1\}$. If the corresponding frustration graph is given by
\begin{equation}
\Gamma = \begin{bmatrix}
0 & -1\\
1 & 0
\end{bmatrix}\oplus \ldots\oplus\begin{bmatrix}
0 & -1\\
1 & 0
\end{bmatrix},
\end{equation}
then there exists a unitary $U: \mathcal{H} \rightarrow \bigotimes_{i=1}^{m}\mathcal{H}_{i}\otimes \mathcal{H}'$ for $\mathcal{H}_{i}=\mathbb{C}^{d}$ and some $\mathcal{H}'$ such that for all $i\in [m]$
\begin{align}
\begin{split}
U M_{2i-1}U^{\dagger}& = X_{i}\otimes\mathbb{1},\quad UM_{2i}U^{\dagger}=Z_{i}\otimes \mathbb{1},
\end{split}
\end{align}
where $X_{i}, Z_{i}$ are generalized Pauli matrices $X, Z$ acting on $\mathcal{H}_{i}$ and $\mathbb{1}$ act on $\bigotimes_{j\neq i}\mathcal{H}_{j}\otimes \mathcal{H}'$.
\end{lem}
See Appendix \ref{app:theorem_1} for the proof. Notice that in order to relate this lemma to the rest of the considerations in this section, we slightly abuse the definition of the frustration graph as the set of matrices $M_{i}$ considered in Lemma \ref{lem:self-testing} do not necessarily form a group. However, since the group property has no impact on Eq. \eqref{eq:Gamma_def}, the meaning of a frustration graph of the set of matrices $M_{i}$ is well-defined.

We are finally ready to formulate the main result of this section.
\begin{thm}\label{thm:self-testing}
Let $\mathcal{A}=\langle T_{1},\dots, T_{k}\rangle_{\odot}$ be a group as in Definition \ref{def: A} and let $\gamma$ be a generating graph corresponding to $\{T_{i}\}_{i=1}^{k}$. There exists a unitary $U$ such that for every element from $\mathcal{A}$ one has
\begin{equation}\label{eq:thm_s-t}
U A\, U^{\dagger}=P_{A}\otimes C_{A},
\end{equation}
where $C_{A}$ is a unitary matrix such that $C_{A}^{d}=\mathbb{1}$ and $[C_{A},C_{A'}]=0$ for all $A,A'\in\mathcal{A}$, and $P_{A}\in \tilde{\mathbb{P}}_{q/2}$ for $q=\operatorname{rank}(\gamma)$.
\end{thm}
\begin{proof}
Before presenting the main ideas of the proof, whose full version can be found in Appendix \ref{app:theorem_1}, let us recall here that $\tilde{\mathbb{P}}_{q/2}$ is the set of all $q/2$-fold tensor products of the $X^{i}Z^{j}$ operators up to a prefactor $\mu_{i,j}\in\{1,\mathbb{i}\}$. Thus, the above theorem allows one to represent all elements of $\mathcal{A}$ as tensor products of $X^{i}Z^{j}$, up to the extra degrees of freedom contained in the $C_A$ operators, which are all diagonal in the same basis.

First, we prove that $\mathcal{C}(\mathcal{A})=d^{\operatorname{null}(\gamma)}$, which implies that we can choose the generating set $\{T_{i}\}_{i=1}^{k}$ such that $\mathcal{S}(\mathcal{A})=\langle T_{1},\ldots,T_{q}\rangle_{\odot}$ with $q=\mathrm{rank}(\gamma)$. Then, from \cite[Lemma 5.3]{Sarkar_2024} it follows that there exists a generating set $\{T_{i}'\}_{i=1}^{q}$ of a subgroup $\mathcal{S}(\mathcal{A})$ for which the adjacency matrix $\gamma'$ of the corresponding generating graph is given by 
\begin{equation}\label{eq:gamma_cannon}
\gamma'=\begin{bmatrix}
0 & -1\\
1 & 0
\end{bmatrix}\oplus \ldots \oplus \begin{bmatrix}
0 & -1\\
1 & 0
\end{bmatrix}.
\end{equation}
Clearly, such generators $T_{i}'$ satisfy the conditions of Lemma \ref{lem:self-testing}, implying the existence of a unitary $U$ such that
\begin{align}
\begin{split}
U T_{2i-1}'U^{\dagger} = X_{i}\otimes\mathbb{1},\qquad UT_{2i}'U^{\dagger}=Z_{i}\otimes\mathbb{1}
\end{split}
\end{align}
for all $i\in [q/2]$. Then, it follows from Eq. \eqref{eq:def_A_I} that for every $A\in\mathcal{S}(\mathcal{A})$ we have that
\begin{equation}
U A U^{\dagger}= P_{A} \otimes \mathbb{1}
\end{equation}
for some $P_{A}\in\tilde{\mathbb{P}}_{q/2}$.

Let us now consider the subgroup $\mathcal{C}(\mathcal{A})$. Since every $A_{C}\in \mathcal{C}(\mathcal{A})$ commutes with every element from the subgroup $\mathcal{S}(\mathcal{A})$ we have that
\begin{equation}
UA_{C}U^{\dagger} = \mathbb{1}_{q/2}\otimes C_{A_{C}},
\end{equation}
where $C_{A_{C}}$ is some unitary matrix satisfying $C_{A_{C}}^{d}=\mathbb{1}$ and $\mathbb{1}_{q/2}$ acts on $\bigotimes_{i=1}^{q/2}\mathcal{H}_{i}=(\mathbb{C}^{d})^{\otimes q/2}$. From the fact that every pair $A_{C},A_{C}'\in \mathcal{C}(\mathcal{A})$ commutes, we have that $[C_{A_{C}},C_{A_{C}}']=0$ for all $C_{A_{C}},C_{A_{C}}'$. 

Lastly, since for every element $A\in\mathcal{A}$ there exist $A_{S}\in \langle T_{1},\dots,T_{q}\rangle_{\odot}$ and $A_{C} \in \mathcal{C}(\mathcal{A})$ such that $A=A_{S}\odot A_{C}$, we can conclude that
\begin{equation}
U A U^{\dagger} = U A_{S} U^{\dagger} U A_{C} U^{\dagger} = P_{A}\otimes C_{A}.
\end{equation}
which ends the proof.
\end{proof}

As a side note, let us mention that Theorem \ref{thm:self-testing} can be extended to sets of $\mathcal{A}$ that are not closed under the operation $\odot$. One simply has to identify a larger group $\tilde{\mathcal{A}}$ such that $\mathcal{A}\subset \tilde{\mathcal{A}}$, then apply Theorem \ref{thm:self-testing} to the elements of $\tilde{\mathcal{A}}$. Afterward, the transformed elements of $\mathcal{A}$ can be taken from the larger set of transformed elements of $\tilde{\mathcal{A}}$.

It is important to note that similar structures have been studied before. In Ref. \cite{978-0-7923-0703-7} it was proven that quasi-Clifford algebras have a unique representation in the matrix algebra. This representation has a very similar structure to the operators Eq. \eqref{eq:thm_s-t} for the case of $d=2$. That is no coincidence since $\mathcal{A}$ with operation $\odot$ can also be viewed as an example of a generalized quasi-Clifford algebra (in our formalism, generalized refers to $d\geqslant 2$). We want to stress, however, that even though our result may hint that the results of Ref. \cite{978-0-7923-0703-7} can be extended to generalized quasi-Clifford algebras, Theorem \ref{thm:self-testing} does not constitute a proper proof of such an extension.

\subsection*{Applications}

Let us now present several interesting applications of Theorem \ref{thm:self-testing}
in various problems frequently considered in quantum information.

\subsubsection*{Upper bound on a sum of squares over $\mathcal{A}$}
The first application concerns finding a tight upper bound to 
the sum of squares of absolute values of expected values over all elements in a group $\mathcal{A}$:
\begin{equation}\label{eq:sos_without<}
\sum_{A\in\mathcal{A}}|\langle A\rangle|^{2},
\end{equation}
where $\langle \cdot \rangle= \operatorname{Tr}[\cdot \rho]$ for an arbitrary state $\rho$. A similar problem was considered in \cite{PhysRevA.107.062211, Xu_2024} in which such a sum was taken over any set of unitary Hermitian matrices that do not necessarily form a group. This is in contrast to this work, where we assume the group structure; however, our operators are not necessarily Hermitian, yet they equal identity when raised to the power $d$.

In Ref. \cite{PhysRevA.107.062211}, the authors established an upper bound on this expression in terms of the Lovász number of the frustration graph. This topic was then studied more in-depth in Ref. \cite{Xu_2024} where it was shown that the Lovász number constitutes a tight upper bound for Eq. \eqref{eq:sos_without<}, however, only if we relax the commutation assumption, i.e., the assumption stating that for any two operators $A, A'$, $[A,A']_{\bullet}=\pm \mathbb{1}$. Instead, the assumption made is that for some pairs of operators $A, A'$ we have $[A,A']_{\bullet}=- \mathbb{1}$, without specifying the commutation relation of the rest of the operators. 

More importantly, from the perspective of this work, it was conjectured in Ref. \cite{PhysRevA.107.062211} that a clique number $\tilde{\omega}(\overline{G})$ is a tight upper bound on the expression \eqref{eq:sos_without<}. This was ultimately disproven in Ref. \cite{Xu_2024} with a counterexample; however, in that work, some examples of adjacency graphs for which a clique number constitutes a tight upper bound were also identified. Here, inspired by the stabilizer formalism, we focus on studying sets of operators $A$ that form a group and we show that for such a set Eq. \eqref{eq:sos_without<} is in fact bounded from above by $\tilde{\omega}(\overline{G})$.

\begin{thm}\label{thm:sos}
Let $\mathcal{A}=\langle T_{1},\dots,T_{k}\rangle_{\odot}$ be a group as in Definition \ref{def: A}. For each such $\mathcal{A}$ we have
\begin{equation}\label{eq:sos_bound_2}
\sum_{A\in\mathcal{A}}|\langle A\rangle|^{2}\leqslant d^{(\operatorname{null}(\gamma)+k)/2}=\tilde{\omega}(\overline{G}).
\end{equation}
\end{thm}
The detailed proof can be found in Appendix \ref{app:sos}. The main idea of the proof is to rewrite the sum of squares in Eq. (\ref{eq:sos_without<}) as 
\begin{equation}
\sum_{A\in\mathcal{A}}|\langle A\rangle|^{2} = \tr[\sum_{A\in\mathcal{A}}(A\otimes A^{\dagger})\rho\otimes\rho],
\end{equation}
which then allows us to use Theorem \ref{thm:self-testing} to express the term $\sum_{A\in\mathcal{A}}A\otimes A^{\dagger}$ as the tensor product of swap operators $U_{\textrm{swap}}:\ket{a}\ket{b}\rightarrow \ket{b}\ket{a}$, and some ancillary, mutually commuting operators. Then by utilizing the fact that $|\mathcal{C}(\mathcal{A})|=d^{\operatorname{null}(\gamma)}$, we can show that the maximal eigenvalue of $\sum_{A\in\mathcal{A}}A\otimes A^{\dagger}$ equals $d^{(\operatorname{null}(\gamma)+k)/2}$. As the final step, for any $\mathcal{A}$ we construct a state $\ket{\psi}$ such that
\begin{equation}
\bra{\psi}\otimes \bra{\psi}\sum_{A\in\mathcal{A}}A\otimes A^{\dagger}\ket{\psi}\otimes\ket{\psi} = d^{(\operatorname{null}(\gamma)+k)/2},
\end{equation}
showing that this upper bound is saturable.

It is easy to see that this bound is always saturable; after all, $\tilde{\omega}(\overline{G})$ is the cardinality of the largest subset of $\mathcal{A}$ in which all matrices mutually commute. Mutual commutation implies a common eigenbasis, and clearly any eigenvector from this basis saturates \eqref{eq:sos_bound_2}. Moreover, notice that the equivalence in Eq. \eqref{eq:sos_bound_2} can be rewritten as
\begin{equation}\label{eq:gamma_expl}
\tilde{\omega}(\overline{G})=d^{(\operatorname{null}(\gamma)+k)/2}=|\mathcal{C}(\mathcal{A})|d^{\operatorname{rank}(\gamma)/2}.
\end{equation}
where the second equality follows from $|\mathcal{C}(\mathcal{A})|=d^{\operatorname{null}(\gamma)}$. 

This equation gives us a good intuition for the relationship between subgroups $\mathcal{C}(\mathcal{A})$ and $\mathcal{S}(\mathcal{A})$, and $\gamma$. The cardinality of $\mathcal{C}(\mathcal{A})$ is related to $\operatorname{null}(\gamma)$, while $\operatorname{rank}(\gamma)$ determines the cardinality of the largest subset of mutually commuting operators in $\mathcal{S}(\mathcal{A})$. A product between all of the elements of the latter set with all of the elements of $\mathcal{C}(\mathcal{A})$ gives the largest set of mutually commuting operators in $\mathcal{A}$, which is directly implied by Eq. \eqref{eq:gamma_expl}.

As was noticed in Ref. \cite{Xu_2024}, if $\overline{G}$ is a perfect graph, i.e., its clique number $\tilde{\omega}(\overline{G})$ equals its chromatic number $\chi(\overline{G})$, then Eq. \eqref{eq:sos_without<} is upper-bounded by $\tilde{\omega}(\overline{G})$. Since the sum of squares under our assumption is constrained by $\tilde{\omega}(\overline{G})$, one can naturally wonder if this implies that the commutation graphs in our problem are perfect. Surprisingly, that is not the case. Even for a simple example of $\mathcal{A}=\langle X\otimes \mathbb{1}, Z\otimes \mathbb{1}, \mathbb{1}\otimes X, \mathbb{1}\otimes Z\rangle_{\odot}$ one can easily check that $\tilde{\omega}(\overline{G})=4$ while $\chi(\overline{G})=5$, i.e., $\overline{G}$ is not a perfect graph. Therefore, we have identified a new class of graphs for which the clique number is an upper bound on Eq. \eqref{eq:sos_without<}.

To illustrate Theorem \ref{thm:sos} with a simple example, let us consider a group generated by the two Pauli matrices $X$ and $Z$ ($d=2$), $\mathcal{A}_{\mathrm{ex}}:=\langle X,Z\rangle$. Clearly, this group consists of four matrices $\mathbbm{1}$, $X$, $Z$ and $\mathbb{i} XZ$ which is equal to the third Pauli matrix $Y$. Given that the group $\mathcal{A}_{\mathrm{ex}}$ has two generators that anticommute, it is direct to observe that the adjacency matrix is $\gamma=X$. Since $X$ is a full rank matrix, $\mathrm{null}(\gamma)$=0, and therefore Theorem \ref{thm:sos} implies that
\begin{equation}
    \sum_{A\in\mathcal{A}_{\mathrm{ex}}}|\langle A\rangle|^{2}=|\langle \mathbbm{1}\rangle|^{2}+|\langle X\rangle|^{2}+|\langle Z\rangle|^{2}+|\langle Y\rangle|^{2}\leqslant 2.
\end{equation}
The above leads to a well-known inequality for the Pauli matrices, 
\begin{equation}
    |\langle X\rangle|^{2}+|\langle Z\rangle|^{2}+|\langle Y\rangle|^{2}\leqslant 1.
\end{equation}

\subsubsection*{Geometric measure of entanglement for stabilizer subspaces}\label{subsec:GME_measure}
Let us now showcase a utility of Theorem \ref{thm:sos} by calculating the geometric measure of entanglement $E_{\textrm{GM}}^{Q}(\mathcal{V}_{\mathbb{S}})$  for a stabilizer subspace $\mathcal{V}_{\mathbb{S}}$.
\begin{thm}\label{thm:GM}
Let $\mathbb{S}=\expect{g_{1},\dots,g_{k}}$ be a stabilizer with a corresponding stabilizer subspace $\mathcal{V}_{\mathbb{S}}$. Geometric measure of entanglement of $\mathcal{V}_{\mathbb{S}}$ with respect to the bipartition $Q|\overline{Q}$ is given by
\begin{equation}\label{eq bbb}
E_{\textrm{GM}}^{Q}(\mathcal{V}_{\mathbb{S}})=1-d^{-\rank(\gamma_{Q})/2}=1-d^{-k}\tilde{\omega}(\overline{G}_{Q}),
\end{equation}
where $\gamma_{Q}$ is an adjacency matrix of a generating graph corresponding to $\{g_{i}^{(Q)}\}_{i=1}^{k}$, and $\overline{G}_{Q}$ is the commutation graph of $\{s^{(Q)}\}_{s\in\mathbb{S}}$.
\end{thm}
The proof can be found in Appendix \ref{app:entanglement_measure}. The main idea of the proof is to lower-bound $E_{\textrm{GM}}^{Q}(\mathcal{V}_{\mathbb{S}})$ via the Cauchy-Schwarz inequality, the result of which is then evaluated exactly using Theorem \ref{thm:sos}. The last step involves showing that for each $\mathcal{V}_{\mathbb{S}}$, we can saturate this bound, proving the equality.

Interestingly, the matrix $\gamma_{Q}$ from Theorem \ref{thm:GM} is equivalent to a commutation matrix - which is a stabilizer-oriented formalism introduced in Ref. \cite{Englbrecht_2022}. It is also worth pointing out that in Ref. \cite{fattal2004entanglementstabilizerformalism}, a different measure of entanglement was studied with respect to the stabilizer formalism. To compute this measure, one considers a bipartition $Q|\overline{Q}$ and transforms a stabilizer state into a tensor product of Bell pairs and product states (which can always be done in the stabilizer formalism). Then, the entanglement measure across $Q|\overline{Q}$ is equal to the number of these Bell pairs. From the perspective of frustration graph formalism, the number of generalized Bell pairs across $Q|\overline{Q}$ equals $\operatorname{rank}(\gamma_{Q})/2$. This follows from the fact that the Eq. \eqref{eq:gamma_cannon} gives us the desired decomposition into Bell pairs and product states: the number of Bell pairs is the number of blocks in Eq. \eqref{eq:gamma_cannon}. This observation not only allows us to compute the measure from Ref. \cite{fattal2004entanglementstabilizerformalism}, but it also showcases its close relationship to the geometric measure in stabilizer formalism, as Eq. \eqref{eq bbb} explicitly depends on the number of Bell pairs $\operatorname{rank}(\gamma_{Q})/2$.

To see how this theory can be applied in practice, let us consider an example of a five-qudit code \cite{Chau_1997} (for a step-by-step calculation of the following see Appendix \ref{app:five-qudit}). It is defined as a stabilizer subspace associated to $\mathbb{S}_{5}=\langle X\otimes Z\otimes Z\otimes X\otimes \mathbb{1}, \mathbb{1}\otimes X\otimes Z\otimes Z\otimes X, X\otimes \mathbb{1}\otimes X\otimes Z\otimes Z, Z \otimes X \otimes \mathbb{1}\otimes X\otimes Z\rangle$. Since this stabilizer is invariant under the cyclic exchange of qudits, there are only three distinct bipartitions: $Q=\{1\},\{1,2\},\{1,3\}$.

The adjacency matrices for the first two bipartitions are given by 
\begin{equation}
\gamma_{1} = \begin{bmatrix}
0 & 0 & 0 & -1\\
0 & 0 & 0 & 0\\
0 & 0 & 0 & -1\\
1 & 0 & 1 & 0
\end{bmatrix},\; \gamma_{1,2} = \begin{bmatrix}
0 & 1 & 0 & 0\\
-1 & 0 & 0 & 0\\
0 & 0 & 0 & -1\\
0 & 0 & 1 & 0
\end{bmatrix},
\end{equation}
with their rank $\operatorname{rank}(\gamma_{1})=2$ and $\operatorname{rank}(\gamma_{1,2})=4$. One can also easily check that $\operatorname{rank}(\gamma_{1,3})=4$. Since all bipartitions for which $|Q|=1,4$ are equivalent to $Q=\{1\}$ and all bipartitions $|Q|=2,3$ are equivalent to either $Q=\{1,2\}$ or $Q=\{1,3\}$, we have that 
\begin{equation}
E_{\textrm{GM}}^{Q}(\mathcal{V}_{\mathbb{S}_{5}})=1-d^{\min(|Q|,5-|Q|)}.
\end{equation}

Returning to general considerations, since we have derived the expression for the geometric measure for a fixed bipartition $Q|\overline{Q}$, we can now use it to calculate the generalized geometric measure of entanglement for any GME stabilizer subspace. 
\begin{cor}\label{cor:GGM}
Let $\mathbb{S}=\langle g_{1},\dots, g_{k}\rangle$ be a stabilizer such that the corresponding stabilizer subspace $\mathcal{V}_{\mathbb{S}}$ is genuinely multipartite entangled. For any such $\mathcal{V}_{\mathbb{S}}$, the generalized geometric measure of entanglement equals
\begin{equation}\label{Naturlijk}
E_{\textrm{GGM}}(\mathcal{V}_{\mathbb{S}})=\frac{d-1}{d}.
\end{equation}
\end{cor}
We give a detailed proof in Appendix \ref{app:GGM}, but the underlying idea is that for each GME $\mathcal{V}_{\mathbb{S}}$ and for each bipartition $Q|\overline{Q}$ such that $|Q|=1$ we have $E_{GM}^{Q}(\mathcal{V}_{\mathbb{S}})=(d-1)/d$. Since by Theorem \ref{thm:GM} this is the smallest achievable $E_{GM}^{Q}(\mathcal{V}_{\mathbb{S}})$ for any $Q|\overline{Q}$ this implies that $E_{\textrm{GGM}}(\mathcal{V}_{\mathbb{S}})=(d-1)/d$. 

It should be stressed here that (\ref{Naturlijk}) is the highest possible value of $E_{\textrm{GGM}}$ \cite{Chen_2014}, and consequently, all genuinely multipartite entangled stabilizer subspaces are also maximally entangled in this sense. What is more, it follows from Ref. \cite{PhysRevLett.122.120503} that GME stabilizer subspaces must automatically maximize any other measure of genuine entanglement that is monotonic under biseparability-preserving transformations.

Lastly, let us note here that by the results of Ref. \cite{Antipin_2021}, this value can be used to lower-bound concurrence and negativity for genuinely multipartite entangled stabilizer subspaces.

\subsubsection*{Upper bound on a sum over $\mathcal{A}$}\label{subsec:SOEV}
Another use case of our result is calculating the sum of expected values over $\mathcal{A}$ for odd, prime $d$. 
\begin{equation}\label{eq:H}
 \sum_{A\in\mathcal{A}} (\langle A\rangle +\langle A^{\dagger}\rangle).
\end{equation}
Interestingly, since this is a linear problem and the operator above is Hermitian, one can instead formulate this task as calculating a bound on the maximal energy level of a Hamiltonian $H=\sum_{A\in\mathcal{A}}  (A + A^{\dagger})$ (for a more in-depth explanation of this class of problems, see Ref. \cite{hastings2023optimizing}). The upper bound for this sum is given in the following theorem.

\begin{thm}\label{thm:H}
Let $\mathcal{A}=\langle T_{1},\ldots T_{k}\rangle_{\odot}$ be a group as in Definition \ref{def: A} and let $d$ be an odd prime number. For each such $\mathcal{A}$, we have the following saturable upper bound
\begin{equation}
 \sum_{A\in\mathcal{A}} (\langle A\rangle + \langle A^{\dagger}\rangle) \leqslant 2 \tilde{\omega}(\overline{G})\left(\frac{1+\sqrt{d}}{2}\right)^{\rank(\gamma)/2}.
\end{equation}
\end{thm}
The proof can be found in Appendix \ref{app:H}, however, the underlying idea is quite easy. The most important observation is that the subgroup $\mathcal{S}(\mathcal{A})$ can be transformed via Theorem \ref{thm:self-testing} into a product of projectors onto states $1/\sqrt{d}\sum_{j=0}^{d-1}\ket{j}$ and $\ket{0}$. Then the above result follows after some simple calculations.

\section{Conclusion}

In this work, we have studied the properties of a group of operators satisfying three conditions: each operator is unitary, each operator taken to the $d$'th power equals identity, and each pair of operators mutually commute up to a power of $d$'th root of unity.

First, we have shown that for such a group, there exists a single unitary transformation that brings the operators into a tensor product of generalized Pauli operators and some ancillary mutually commuting operators. Then, this allowed us to find an upper bound on the sum of squares of absolute values of the expected values of these operators which amounts to the clique number of a commutation graph of this group, and so we have thus identified another class of graphs, next to perfect graphs, for which the clique number constitutes a proper upper bound. Next, we showed that the bound on the sum of squares can be directly utilized for the computation of the geometric and generalized geometric measures of entanglement for genuinely entangled stabilizer subspaces. Lastly, we also computed an upper bound on the sum of expected values of these operators, which can be interpreted as a derivation of an upper bound on the highest energy level of a certain specific many-body Hamiltonian.

This work leaves many interesting avenues for further exploration.
\begin{itemize}
    \item First, as we mentioned in the text, Theorem \ref{thm:self-testing} hints that one can extend the proof of a unique representation of quasi-Clifford algebras in a given matrix algebra to the case of generalized quasi-Clifford algebras. In our work, we were interested in operators admitting certain conditions, not in general mathematical objects, and so a proof of such a unique representation would require a much more general approach. 

\item The second open problem is to generalize the results from Ref. \cite{Xu_2024} but for the operators that equal identity when taken to the power $d$ and commute up to the $d$'th root of identity, that is, obey the commutation relation (\ref{eq:Gamma_def}). In particular, it would be interesting to see whether the hierarchy of upper bounds derived under different assumptions about the operators (see \cite{Xu_2024}{ Proposition 2}) remains unchanged in the case of higher $d$.

\item Last but not least, one can also explore whether the presented formalism can be modified to compute the geometric measure of entanglement for genuinely entangled subspaces, where the maximization is performed over fully product states, rather than states product across a given bipartition.
It would be intriguing to explore whether the geometric measure of entanglement behaves similarly to the generalized geometric measure of entanglement, remaining constant for all such subspaces within a given local dimension.
\end{itemize}

\section*{Data availability statement}
No data set was generated or analyzed in the current study.

\section*{Competing interests}
The Authors declare no Competing Financial or Non-Financial Interests.

\begin{acknowledgments}
We thank Błażej Ruba, Carlos de Gois, Kiara Hansenne, and Ignacy Stachura for insightful discussions. We are also grateful to Samuel Elman and Julio de Vicente for bringing Refs. \cite{mann2024graphtheoreticframeworkfreeparafermionsolvability} and \cite{PhysRevLett.122.120503}, respectively, to our attention. This work is supported by the National Science Centre (Poland) through the SONATA BIS project No. 2019/34/E/ST2/00369. This project has received funding from the European Union’s Horizon Europe research and innovation programme under grant agreement No 101080086 NeQST.
\end{acknowledgments}

%\bibliographystyle{revtex4-2}
%\bibliography{bibliography}

\begin{thebibliography}{45}%
\makeatletter
\providecommand \@ifxundefined [1]{%
 \@ifx{#1\undefined}
}%
\providecommand \@ifnum [1]{%
 \ifnum #1\expandafter \@firstoftwo
 \else \expandafter \@secondoftwo
 \fi
}%
\providecommand \@ifx [1]{%
 \ifx #1\expandafter \@firstoftwo
 \else \expandafter \@secondoftwo
 \fi
}%
\providecommand \natexlab [1]{#1}%
\providecommand \enquote  [1]{``#1''}%
\providecommand \bibnamefont  [1]{#1}%
\providecommand \bibfnamefont [1]{#1}%
\providecommand \citenamefont [1]{#1}%
\providecommand \href@noop [0]{\@secondoftwo}%
\providecommand \href [0]{\begingroup \@sanitize@url \@href}%
\providecommand \@href[1]{\@@startlink{#1}\@@href}%
\providecommand \@@href[1]{\endgroup#1\@@endlink}%
\providecommand \@sanitize@url [0]{\catcode `\\12\catcode `\$12\catcode
  `\&12\catcode `\#12\catcode `\^12\catcode `\_12\catcode `\%12\relax}%
\providecommand \@@startlink[1]{}%
\providecommand \@@endlink[0]{}%
\providecommand \url  [0]{\begingroup\@sanitize@url \@url }%
\providecommand \@url [1]{\endgroup\@href {#1}{\urlprefix }}%
\providecommand \urlprefix  [0]{URL }%
\providecommand \Eprint [0]{\href }%
\providecommand \doibase [0]{https://doi.org/}%
\providecommand \selectlanguage [0]{\@gobble}%
\providecommand \bibinfo  [0]{\@secondoftwo}%
\providecommand \bibfield  [0]{\@secondoftwo}%
\providecommand \translation [1]{[#1]}%
\providecommand \BibitemOpen [0]{}%
\providecommand \bibitemStop [0]{}%
\providecommand \bibitemNoStop [0]{.\EOS\space}%
\providecommand \EOS [0]{\spacefactor3000\relax}%
\providecommand \BibitemShut  [1]{\csname bibitem#1\endcsname}%
\let\auto@bib@innerbib\@empty
%</preamble>
\bibitem [{\citenamefont {Quintino}\ \emph {et~al.}(2014)\citenamefont
  {Quintino}, \citenamefont {V\'ertesi},\ and\ \citenamefont
  {Brunner}}]{PhysRevLett.113.160402}%
  \BibitemOpen
  \bibfield  {author} {\bibinfo {author} {\bibfnamefont {M.~T.}\ \bibnamefont
  {Quintino}}, \bibinfo {author} {\bibfnamefont {T.}~\bibnamefont
  {V\'ertesi}},\ and\ \bibinfo {author} {\bibfnamefont {N.}~\bibnamefont
  {Brunner}},\ }\bibfield  {title} {\bibinfo {title} {Joint measurability,
  einstein-podolsky-rosen steering, and bell nonlocality},\ }\href
  {https://doi.org/10.1103/PhysRevLett.113.160402} {\bibfield  {journal}
  {\bibinfo  {journal} {Phys. Rev. Lett.}\ }\textbf {\bibinfo {volume} {113}},\
  \bibinfo {pages} {160402} (\bibinfo {year} {2014})}\BibitemShut {NoStop}%
\bibitem [{\citenamefont {Mayers}\ and\ \citenamefont {Yao}(1998)}]{743501}%
  \BibitemOpen
  \bibfield  {author} {\bibinfo {author} {\bibfnamefont {D.}~\bibnamefont
  {Mayers}}\ and\ \bibinfo {author} {\bibfnamefont {A.}~\bibnamefont {Yao}},\
  }\bibfield  {title} {\bibinfo {title} {Quantum cryptography with imperfect
  apparatus},\ }in\ \href {https://doi.org/10.1109/SFCS.1998.743501} {\emph
  {\bibinfo {booktitle} {Proceedings 39th Annual Symposium on Foundations of
  Computer Science (Cat. No.98CB36280)}}}\ (\bibinfo {year} {1998})\ pp.\
  \bibinfo {pages} {503--509}\BibitemShut {NoStop}%
\bibitem [{\citenamefont {Zapatero}\ \emph {et~al.}(2023)\citenamefont
  {Zapatero}, \citenamefont {van Leent}, \citenamefont {Arnon-Friedman},
  \citenamefont {Liu}, \citenamefont {Zhang}, \citenamefont {Weinfurter},\ and\
  \citenamefont {Curty}}]{Zapatero2023-kf}%
  \BibitemOpen
  \bibfield  {author} {\bibinfo {author} {\bibfnamefont {V.}~\bibnamefont
  {Zapatero}}, \bibinfo {author} {\bibfnamefont {T.}~\bibnamefont {van Leent}},
  \bibinfo {author} {\bibfnamefont {R.}~\bibnamefont {Arnon-Friedman}},
  \bibinfo {author} {\bibfnamefont {W.-Z.}\ \bibnamefont {Liu}}, \bibinfo
  {author} {\bibfnamefont {Q.}~\bibnamefont {Zhang}}, \bibinfo {author}
  {\bibfnamefont {H.}~\bibnamefont {Weinfurter}},\ and\ \bibinfo {author}
  {\bibfnamefont {M.}~\bibnamefont {Curty}},\ }\bibfield  {title} {\bibinfo
  {title} {Advances in device-independent quantum key distribution},\ }\href
  {https://www.nature.com/articles/s41534-023-00684-x} {\bibfield  {journal}
  {\bibinfo  {journal} {npj Quantum Inf.}\ }\textbf {\bibinfo {volume} {9}},\
  \bibinfo {pages} {10} (\bibinfo {year} {2023})}\BibitemShut {NoStop}%
\bibitem [{\citenamefont {Nadlinger}\ \emph {et~al.}(2022)\citenamefont
  {Nadlinger}, \citenamefont {Drmota}, \citenamefont {Nichol}, \citenamefont
  {Araneda}, \citenamefont {Main}, \citenamefont {Srinivas}, \citenamefont
  {Lucas}, \citenamefont {Ballance}, \citenamefont {Ivanov}, \citenamefont
  {Tan}, \citenamefont {Sekatski}, \citenamefont {Urbanke}, \citenamefont
  {Renner}, \citenamefont {Sangouard},\ and\ \citenamefont
  {Bancal}}]{Nadlinger_2022}%
  \BibitemOpen
  \bibfield  {author} {\bibinfo {author} {\bibfnamefont {D.~P.}\ \bibnamefont
  {Nadlinger}}, \bibinfo {author} {\bibfnamefont {P.}~\bibnamefont {Drmota}},
  \bibinfo {author} {\bibfnamefont {B.~C.}\ \bibnamefont {Nichol}}, \bibinfo
  {author} {\bibfnamefont {G.}~\bibnamefont {Araneda}}, \bibinfo {author}
  {\bibfnamefont {D.}~\bibnamefont {Main}}, \bibinfo {author} {\bibfnamefont
  {R.}~\bibnamefont {Srinivas}}, \bibinfo {author} {\bibfnamefont {D.~M.}\
  \bibnamefont {Lucas}}, \bibinfo {author} {\bibfnamefont {C.~J.}\ \bibnamefont
  {Ballance}}, \bibinfo {author} {\bibfnamefont {K.}~\bibnamefont {Ivanov}},
  \bibinfo {author} {\bibfnamefont {E.~Y.-Z.}\ \bibnamefont {Tan}}, \bibinfo
  {author} {\bibfnamefont {P.}~\bibnamefont {Sekatski}}, \bibinfo {author}
  {\bibfnamefont {R.~L.}\ \bibnamefont {Urbanke}}, \bibinfo {author}
  {\bibfnamefont {R.}~\bibnamefont {Renner}}, \bibinfo {author} {\bibfnamefont
  {N.}~\bibnamefont {Sangouard}},\ and\ \bibinfo {author} {\bibfnamefont
  {J.-D.}\ \bibnamefont {Bancal}},\ }\bibfield  {title} {\bibinfo {title}
  {Experimental quantum key distribution certified by bell’s theorem},\
  }\href {https://doi.org/10.1038/s41586-022-04941-5} {\bibfield  {journal}
  {\bibinfo  {journal} {Nature}\ }\textbf {\bibinfo {volume} {607}},\ \bibinfo
  {pages} {682–686} (\bibinfo {year} {2022})}\BibitemShut {NoStop}%
\bibitem [{\citenamefont {Šupić}\ and\ \citenamefont
  {Bowles}(2020)}]{_upi__2020}%
  \BibitemOpen
  \bibfield  {author} {\bibinfo {author} {\bibfnamefont {I.}~\bibnamefont
  {Šupić}}\ and\ \bibinfo {author} {\bibfnamefont {J.}~\bibnamefont
  {Bowles}},\ }\bibfield  {title} {\bibinfo {title} {Self-testing of quantum
  systems: a review},\ }\href {https://doi.org/10.22331/q-2020-09-30-337}
  {\bibfield  {journal} {\bibinfo  {journal} {Quantum}\ }\textbf {\bibinfo
  {volume} {4}},\ \bibinfo {pages} {337} (\bibinfo {year} {2020})}\BibitemShut
  {NoStop}%
\bibitem [{\citenamefont {Baccari}\ \emph {et~al.}(2020)\citenamefont
  {Baccari}, \citenamefont {Augusiak}, \citenamefont {\ifmmode \check{S}\else
  \v{S}\fi{}upi\ifmmode~\acute{c}\else \'{c}\fi{}},\ and\ \citenamefont
  {Ac\'{\i}n}}]{PhysRevLett.125.260507}%
  \BibitemOpen
  \bibfield  {author} {\bibinfo {author} {\bibfnamefont {F.}~\bibnamefont
  {Baccari}}, \bibinfo {author} {\bibfnamefont {R.}~\bibnamefont {Augusiak}},
  \bibinfo {author} {\bibfnamefont {I.}~\bibnamefont {\ifmmode \check{S}\else
  \v{S}\fi{}upi\ifmmode~\acute{c}\else \'{c}\fi{}}},\ and\ \bibinfo {author}
  {\bibfnamefont {A.}~\bibnamefont {Ac\'{\i}n}},\ }\bibfield  {title} {\bibinfo
  {title} {Device-independent certification of genuinely entangled subspaces},\
  }\href {https://doi.org/10.1103/PhysRevLett.125.260507} {\bibfield  {journal}
  {\bibinfo  {journal} {Phys. Rev. Lett.}\ }\textbf {\bibinfo {volume} {125}},\
  \bibinfo {pages} {260507} (\bibinfo {year} {2020})}\BibitemShut {NoStop}%
\bibitem [{\citenamefont {Ac{\'\i}n}\ and\ \citenamefont
  {Masanes}(2016)}]{Acin2016-xj}%
  \BibitemOpen
  \bibfield  {author} {\bibinfo {author} {\bibfnamefont {A.}~\bibnamefont
  {Ac{\'\i}n}}\ and\ \bibinfo {author} {\bibfnamefont {L.}~\bibnamefont
  {Masanes}},\ }\bibfield  {title} {\bibinfo {title} {Certified randomness in
  quantum physics},\ }\href {https://www.nature.com/articles/nature20119}
  {\bibfield  {journal} {\bibinfo  {journal} {Nature}\ }\textbf {\bibinfo
  {volume} {540}},\ \bibinfo {pages} {213} (\bibinfo {year}
  {2016})}\BibitemShut {NoStop}%
\bibitem [{\citenamefont {Pironio}\ \emph {et~al.}(2010)\citenamefont
  {Pironio}, \citenamefont {Ac{\'\i}n}, \citenamefont {Massar}, \citenamefont
  {de~la Giroday}, \citenamefont {Matsukevich}, \citenamefont {Maunz},
  \citenamefont {Olmschenk}, \citenamefont {Hayes}, \citenamefont {Luo},
  \citenamefont {Manning},\ and\ \citenamefont {Monroe}}]{Pironio2010-ep}%
  \BibitemOpen
  \bibfield  {author} {\bibinfo {author} {\bibfnamefont {S.}~\bibnamefont
  {Pironio}}, \bibinfo {author} {\bibfnamefont {A.}~\bibnamefont {Ac{\'\i}n}},
  \bibinfo {author} {\bibfnamefont {S.}~\bibnamefont {Massar}}, \bibinfo
  {author} {\bibfnamefont {A.~B.}\ \bibnamefont {de~la Giroday}}, \bibinfo
  {author} {\bibfnamefont {D.~N.}\ \bibnamefont {Matsukevich}}, \bibinfo
  {author} {\bibfnamefont {P.}~\bibnamefont {Maunz}}, \bibinfo {author}
  {\bibfnamefont {S.}~\bibnamefont {Olmschenk}}, \bibinfo {author}
  {\bibfnamefont {D.}~\bibnamefont {Hayes}}, \bibinfo {author} {\bibfnamefont
  {L.}~\bibnamefont {Luo}}, \bibinfo {author} {\bibfnamefont {T.~A.}\
  \bibnamefont {Manning}},\ and\ \bibinfo {author} {\bibfnamefont
  {C.}~\bibnamefont {Monroe}},\ }\bibfield  {title} {\bibinfo {title} {Random
  numbers certified by bell's theorem},\ }\href
  {https://www.nature.com/articles/nature09008} {\bibfield  {journal} {\bibinfo
   {journal} {Nature}\ }\textbf {\bibinfo {volume} {464}},\ \bibinfo {pages}
  {1021} (\bibinfo {year} {2010})}\BibitemShut {NoStop}%
\bibitem [{\citenamefont {Maccone}\ and\ \citenamefont
  {Pati}(2014)}]{PhysRevLett.113.260401}%
  \BibitemOpen
  \bibfield  {author} {\bibinfo {author} {\bibfnamefont {L.}~\bibnamefont
  {Maccone}}\ and\ \bibinfo {author} {\bibfnamefont {A.~K.}\ \bibnamefont
  {Pati}},\ }\bibfield  {title} {\bibinfo {title} {Stronger uncertainty
  relations for all incompatible observables},\ }\href
  {https://doi.org/10.1103/PhysRevLett.113.260401} {\bibfield  {journal}
  {\bibinfo  {journal} {Phys. Rev. Lett.}\ }\textbf {\bibinfo {volume} {113}},\
  \bibinfo {pages} {260401} (\bibinfo {year} {2014})}\BibitemShut {NoStop}%
\bibitem [{\citenamefont {de~Guise}\ \emph {et~al.}(2018)\citenamefont
  {de~Guise}, \citenamefont {Maccone}, \citenamefont {Sanders},\ and\
  \citenamefont {Shukla}}]{PhysRevA.98.042121}%
  \BibitemOpen
  \bibfield  {author} {\bibinfo {author} {\bibfnamefont {H.}~\bibnamefont
  {de~Guise}}, \bibinfo {author} {\bibfnamefont {L.}~\bibnamefont {Maccone}},
  \bibinfo {author} {\bibfnamefont {B.~C.}\ \bibnamefont {Sanders}},\ and\
  \bibinfo {author} {\bibfnamefont {N.}~\bibnamefont {Shukla}},\ }\bibfield
  {title} {\bibinfo {title} {State-independent uncertainty relations},\ }\href
  {https://doi.org/10.1103/PhysRevA.98.042121} {\bibfield  {journal} {\bibinfo
  {journal} {Phys. Rev. A}\ }\textbf {\bibinfo {volume} {98}},\ \bibinfo
  {pages} {042121} (\bibinfo {year} {2018})}\BibitemShut {NoStop}%
\bibitem [{\citenamefont {Xiao}\ \emph {et~al.}(2019)\citenamefont {Xiao},
  \citenamefont {Guo}, \citenamefont {Meng}, \citenamefont {Jing},\ and\
  \citenamefont {Yung}}]{PhysRevA.100.032118}%
  \BibitemOpen
  \bibfield  {author} {\bibinfo {author} {\bibfnamefont {Y.}~\bibnamefont
  {Xiao}}, \bibinfo {author} {\bibfnamefont {C.}~\bibnamefont {Guo}}, \bibinfo
  {author} {\bibfnamefont {F.}~\bibnamefont {Meng}}, \bibinfo {author}
  {\bibfnamefont {N.}~\bibnamefont {Jing}},\ and\ \bibinfo {author}
  {\bibfnamefont {M.-H.}\ \bibnamefont {Yung}},\ }\bibfield  {title} {\bibinfo
  {title} {Incompatibility of observables as state-independent bound of
  uncertainty relations},\ }\href {https://doi.org/10.1103/PhysRevA.100.032118}
  {\bibfield  {journal} {\bibinfo  {journal} {Phys. Rev. A}\ }\textbf {\bibinfo
  {volume} {100}},\ \bibinfo {pages} {032118} (\bibinfo {year}
  {2019})}\BibitemShut {NoStop}%
\bibitem [{\citenamefont {Diep}(2024)}]{diep2024frustratedspinsystemshistory}%
  \BibitemOpen
  \bibfield  {author} {\bibinfo {author} {\bibfnamefont {H.~T.}\ \bibnamefont
  {Diep}},\ }\href {https://arxiv.org/abs/2411.12826} {\bibinfo {title}
  {Frustrated spin systems: History of the emergence of a modern physics}}
  (\bibinfo {year} {2024}),\ \Eprint {https://arxiv.org/abs/2411.12826}
  {arXiv:2411.12826 [cond-mat.stat-mech]} \BibitemShut {NoStop}%
\bibitem [{\citenamefont {Hastings}\ and\ \citenamefont
  {O'Donnell}(2022)}]{10.1145/3519935.3519960}%
  \BibitemOpen
  \bibfield  {author} {\bibinfo {author} {\bibfnamefont {M.~B.}\ \bibnamefont
  {Hastings}}\ and\ \bibinfo {author} {\bibfnamefont {R.}~\bibnamefont
  {O'Donnell}},\ }\bibfield  {title} {\bibinfo {title} {Optimizing strongly
  interacting fermionic hamiltonians},\ }in\ \href
  {https://doi.org/10.1145/3519935.3519960} {\emph {\bibinfo {booktitle}
  {Proceedings of the 54th Annual ACM SIGACT Symposium on Theory of
  Computing}}},\ \bibinfo {series and number} {STOC 2022}\ (\bibinfo
  {publisher} {Association for Computing Machinery},\ \bibinfo {address} {New
  York, NY, USA},\ \bibinfo {year} {2022})\ p.\ \bibinfo {pages}
  {776–789}\BibitemShut {NoStop}%
\bibitem [{\citenamefont {Chapman}\ and\ \citenamefont
  {Flammia}(2020)}]{Chapman2020characterizationof}%
  \BibitemOpen
  \bibfield  {author} {\bibinfo {author} {\bibfnamefont {A.}~\bibnamefont
  {Chapman}}\ and\ \bibinfo {author} {\bibfnamefont {S.~T.}\ \bibnamefont
  {Flammia}},\ }\bibfield  {title} {\bibinfo {title} {Characterization of
  solvable spin models via graph invariants},\ }\href
  {https://doi.org/10.22331/q-2020-06-04-278} {\bibfield  {journal} {\bibinfo
  {journal} {{Quantum}}\ }\textbf {\bibinfo {volume} {4}},\ \bibinfo {pages}
  {278} (\bibinfo {year} {2020})}\BibitemShut {NoStop}%
\bibitem [{\citenamefont {Chapman}\ \emph {et~al.}(2023)\citenamefont
  {Chapman}, \citenamefont {Elman},\ and\ \citenamefont
  {Mann}}]{chapman2023unifiedgraphtheoreticframeworkfreefermion}%
  \BibitemOpen
  \bibfield  {author} {\bibinfo {author} {\bibfnamefont {A.}~\bibnamefont
  {Chapman}}, \bibinfo {author} {\bibfnamefont {S.~J.}\ \bibnamefont {Elman}},\
  and\ \bibinfo {author} {\bibfnamefont {R.~L.}\ \bibnamefont {Mann}},\ }\href
  {https://arxiv.org/abs/2305.15625} {\bibinfo {title} {A unified
  graph-theoretic framework for free-fermion solvability}} (\bibinfo {year}
  {2023}),\ \Eprint {https://arxiv.org/abs/2305.15625} {arXiv:2305.15625
  [quant-ph]} \BibitemShut {NoStop}%
\bibitem [{\citenamefont {de~Gois}\ \emph {et~al.}(2023)\citenamefont
  {de~Gois}, \citenamefont {Hansenne},\ and\ \citenamefont
  {G\"uhne}}]{PhysRevA.107.062211}%
  \BibitemOpen
  \bibfield  {author} {\bibinfo {author} {\bibfnamefont {C.}~\bibnamefont
  {de~Gois}}, \bibinfo {author} {\bibfnamefont {K.}~\bibnamefont {Hansenne}},\
  and\ \bibinfo {author} {\bibfnamefont {O.}~\bibnamefont {G\"uhne}},\
  }\bibfield  {title} {\bibinfo {title} {Uncertainty relations from graph
  theory},\ }\href {https://doi.org/10.1103/PhysRevA.107.062211} {\bibfield
  {journal} {\bibinfo  {journal} {Phys. Rev. A}\ }\textbf {\bibinfo {volume}
  {107}},\ \bibinfo {pages} {062211} (\bibinfo {year} {2023})}\BibitemShut
  {NoStop}%
\bibitem [{\citenamefont {Xu}\ \emph {et~al.}(2024)\citenamefont {Xu},
  \citenamefont {Schwonnek},\ and\ \citenamefont {Winter}}]{Xu_2024}%
  \BibitemOpen
  \bibfield  {author} {\bibinfo {author} {\bibfnamefont {Z.-P.}\ \bibnamefont
  {Xu}}, \bibinfo {author} {\bibfnamefont {R.}~\bibnamefont {Schwonnek}},\ and\
  \bibinfo {author} {\bibfnamefont {A.}~\bibnamefont {Winter}},\ }\bibfield
  {title} {\bibinfo {title} {Bounding the joint numerical range of pauli
  strings by graph parameters},\ }\href
  {http://dx.doi.org/10.1103/PRXQuantum.5.020318} {\bibfield  {journal}
  {\bibinfo  {journal} {PRX Quantum}\ }\textbf {\bibinfo {volume} {5}}
  (\bibinfo {year} {2024})}\BibitemShut {NoStop}%
\bibitem [{\citenamefont {Aguilar}\ \emph {et~al.}(2024)\citenamefont
  {Aguilar}, \citenamefont {Cichy}, \citenamefont {Eisert},\ and\ \citenamefont
  {Bittel}}]{aguilar2024classificationpauliliealgebras}%
  \BibitemOpen
  \bibfield  {author} {\bibinfo {author} {\bibfnamefont {G.}~\bibnamefont
  {Aguilar}}, \bibinfo {author} {\bibfnamefont {S.}~\bibnamefont {Cichy}},
  \bibinfo {author} {\bibfnamefont {J.}~\bibnamefont {Eisert}},\ and\ \bibinfo
  {author} {\bibfnamefont {L.}~\bibnamefont {Bittel}},\ }\href
  {https://arxiv.org/abs/2408.00081} {\bibinfo {title} {Full classification of
  pauli lie algebras}} (\bibinfo {year} {2024}),\ \Eprint
  {https://arxiv.org/abs/2408.00081} {arXiv:2408.00081 [quant-ph]} \BibitemShut
  {NoStop}%
\bibitem [{\citenamefont {Mor\'an}\ and\ \citenamefont
  {Huber}(2024)}]{Mor_n_2024}%
  \BibitemOpen
  \bibfield  {author} {\bibinfo {author} {\bibfnamefont {M.~B.}\ \bibnamefont
  {Mor\'an}}\ and\ \bibinfo {author} {\bibfnamefont {F.}~\bibnamefont
  {Huber}},\ }\bibfield  {title} {\bibinfo {title} {Uncertainty relations from
  state polynomial optimization},\ }\href
  {https://doi.org/10.1103/PhysRevLett.132.200202} {\bibfield  {journal}
  {\bibinfo  {journal} {Phys. Rev. Lett.}\ }\textbf {\bibinfo {volume} {132}},\
  \bibinfo {pages} {200202} (\bibinfo {year} {2024})}\BibitemShut {NoStop}%
\bibitem [{\citenamefont {Hansenne}\ \emph {et~al.}(2022)\citenamefont
  {Hansenne}, \citenamefont {Xu}, \citenamefont {Kraft},\ and\ \citenamefont
  {Gühne}}]{Hansenne_2022}%
  \BibitemOpen
  \bibfield  {author} {\bibinfo {author} {\bibfnamefont {K.}~\bibnamefont
  {Hansenne}}, \bibinfo {author} {\bibfnamefont {Z.-P.}\ \bibnamefont {Xu}},
  \bibinfo {author} {\bibfnamefont {T.}~\bibnamefont {Kraft}},\ and\ \bibinfo
  {author} {\bibfnamefont {O.}~\bibnamefont {Gühne}},\ }\bibfield  {title}
  {\bibinfo {title} {Symmetries in quantum networks lead to no-go theorems for
  entanglement distribution and to verification techniques},\ }\href
  {http://dx.doi.org/10.1038/s41467-022-28006-3} {\bibfield  {journal}
  {\bibinfo  {journal} {Nat. Comm.}\ }\textbf {\bibinfo {volume} {13}},\
  \bibinfo {pages} {496} (\bibinfo {year} {2022})}\BibitemShut {NoStop}%
\bibitem [{\citenamefont {Mann}\ \emph {et~al.}(2024)\citenamefont {Mann},
  \citenamefont {Elman}, \citenamefont {Wood},\ and\ \citenamefont
  {Chapman}}]{mann2024graphtheoreticframeworkfreeparafermionsolvability}%
  \BibitemOpen
  \bibfield  {author} {\bibinfo {author} {\bibfnamefont {R.~L.}\ \bibnamefont
  {Mann}}, \bibinfo {author} {\bibfnamefont {S.~J.}\ \bibnamefont {Elman}},
  \bibinfo {author} {\bibfnamefont {D.~R.}\ \bibnamefont {Wood}},\ and\
  \bibinfo {author} {\bibfnamefont {A.}~\bibnamefont {Chapman}},\ }\href
  {https://arxiv.org/abs/2408.09684} {\bibinfo {title} {A graph-theoretic
  framework for free-parafermion solvability}} (\bibinfo {year} {2024}),\
  \Eprint {https://arxiv.org/abs/2408.09684} {arXiv:2408.09684 [quant-ph]}
  \BibitemShut {NoStop}%
\bibitem [{\citenamefont {Sarkar}\ and\ \citenamefont
  {Yoder}(2024)}]{Sarkar_2024}%
  \BibitemOpen
  \bibfield  {author} {\bibinfo {author} {\bibfnamefont {R.}~\bibnamefont
  {Sarkar}}\ and\ \bibinfo {author} {\bibfnamefont {T.~J.}\ \bibnamefont
  {Yoder}},\ }\bibfield  {title} {\bibinfo {title} {The qudit pauli group:
  non-commuting pairs, non-commuting sets, and structure theorems},\ }\href
  {https://doi.org/10.22331/q-2024-04-04-1307} {\bibfield  {journal} {\bibinfo
  {journal} {Quantum}\ }\textbf {\bibinfo {volume} {8}},\ \bibinfo {pages}
  {1307} (\bibinfo {year} {2024})}\BibitemShut {NoStop}%
\bibitem [{\citenamefont {Santos}\ \emph {et~al.}(2022)\citenamefont {Santos},
  \citenamefont {Jebarathinam},\ and\ \citenamefont
  {Augusiak}}]{PhysRevA.106.012431}%
  \BibitemOpen
  \bibfield  {author} {\bibinfo {author} {\bibfnamefont {R.}~\bibnamefont
  {Santos}}, \bibinfo {author} {\bibfnamefont {C.}~\bibnamefont
  {Jebarathinam}},\ and\ \bibinfo {author} {\bibfnamefont {R.}~\bibnamefont
  {Augusiak}},\ }\bibfield  {title} {\bibinfo {title} {Scalable
  noncontextuality inequalities and certification of multiqubit quantum
  systems},\ }\href {https://doi.org/10.1103/PhysRevA.106.012431} {\bibfield
  {journal} {\bibinfo  {journal} {Phys. Rev. A}\ }\textbf {\bibinfo {volume}
  {106}},\ \bibinfo {pages} {012431} (\bibinfo {year} {2022})}\BibitemShut
  {NoStop}%
\bibitem [{\citenamefont {Steane}(1996)}]{PhysRevLett.77.793}%
  \BibitemOpen
  \bibfield  {author} {\bibinfo {author} {\bibfnamefont {A.~M.}\ \bibnamefont
  {Steane}},\ }\bibfield  {title} {\bibinfo {title} {Error correcting codes in
  quantum theory},\ }\href {https://doi.org/10.1103/PhysRevLett.77.793}
  {\bibfield  {journal} {\bibinfo  {journal} {Phys. Rev. Lett.}\ }\textbf
  {\bibinfo {volume} {77}},\ \bibinfo {pages} {793} (\bibinfo {year}
  {1996})}\BibitemShut {NoStop}%
\bibitem [{\citenamefont {Nadkarni}\ and\ \citenamefont
  {Garani}(2021)}]{PhysRevA.103.042420}%
  \BibitemOpen
  \bibfield  {author} {\bibinfo {author} {\bibfnamefont {P.~J.}\ \bibnamefont
  {Nadkarni}}\ and\ \bibinfo {author} {\bibfnamefont {S.~S.}\ \bibnamefont
  {Garani}},\ }\bibfield  {title} {\bibinfo {title} {Quantum error correction
  architecture for qudit stabilizer codes},\ }\href
  {https://doi.org/10.1103/PhysRevA.103.042420} {\bibfield  {journal} {\bibinfo
   {journal} {Phys. Rev. A}\ }\textbf {\bibinfo {volume} {103}},\ \bibinfo
  {pages} {042420} (\bibinfo {year} {2021})}\BibitemShut {NoStop}%
\bibitem [{\citenamefont
  {Gottesman}(1997)}]{gottesman1997stabilizercodesquantumerror}%
  \BibitemOpen
  \bibfield  {author} {\bibinfo {author} {\bibfnamefont {D.}~\bibnamefont
  {Gottesman}},\ }\href {https://arxiv.org/abs/quant-ph/9705052} {\bibinfo
  {title} {Stabilizer codes and quantum error correction}} (\bibinfo {year}
  {1997}),\ \Eprint {https://arxiv.org/abs/quant-ph/9705052}
  {arXiv:quant-ph/9705052 [quant-ph]} \BibitemShut {NoStop}%
\bibitem [{\citenamefont {T\'oth}\ and\ \citenamefont
  {G\"uhne}(2005)}]{PhysRevA.72.022340}%
  \BibitemOpen
  \bibfield  {author} {\bibinfo {author} {\bibfnamefont {G.}~\bibnamefont
  {T\'oth}}\ and\ \bibinfo {author} {\bibfnamefont {O.}~\bibnamefont
  {G\"uhne}},\ }\bibfield  {title} {\bibinfo {title} {Entanglement detection in
  the stabilizer formalism},\ }\href
  {https://doi.org/10.1103/PhysRevA.72.022340} {\bibfield  {journal} {\bibinfo
  {journal} {Phys. Rev. A}\ }\textbf {\bibinfo {volume} {72}},\ \bibinfo
  {pages} {022340} (\bibinfo {year} {2005})}\BibitemShut {NoStop}%
\bibitem [{\citenamefont {Makuta}\ and\ \citenamefont
  {Augusiak}(2021)}]{Makuta_2021}%
  \BibitemOpen
  \bibfield  {author} {\bibinfo {author} {\bibfnamefont {O.}~\bibnamefont
  {Makuta}}\ and\ \bibinfo {author} {\bibfnamefont {R.}~\bibnamefont
  {Augusiak}},\ }\bibfield  {title} {\bibinfo {title} {Self-testing
  maximally-dimensional genuinely entangled subspaces within the stabilizer
  formalism},\ }\href {https://doi.org/10.1088/1367-2630/abee40} {\bibfield
  {journal} {\bibinfo  {journal} {New J. Phys.}\ }\textbf {\bibinfo {volume}
  {23}},\ \bibinfo {pages} {043042} (\bibinfo {year} {2021})}\BibitemShut
  {NoStop}%
\bibitem [{\citenamefont {Sen(De)}\ and\ \citenamefont
  {Sen}(2010)}]{PhysRevA.81.012308}%
  \BibitemOpen
  \bibfield  {author} {\bibinfo {author} {\bibfnamefont {A.}~\bibnamefont
  {Sen(De)}}\ and\ \bibinfo {author} {\bibfnamefont {U.}~\bibnamefont {Sen}},\
  }\bibfield  {title} {\bibinfo {title} {Channel capacities versus entanglement
  measures in multiparty quantum states},\ }\href
  {https://doi.org/10.1103/PhysRevA.81.012308} {\bibfield  {journal} {\bibinfo
  {journal} {Phys. Rev. A}\ }\textbf {\bibinfo {volume} {81}},\ \bibinfo
  {pages} {012308} (\bibinfo {year} {2010})}\BibitemShut {NoStop}%
\bibitem [{\citenamefont {Seevinck}\ and\ \citenamefont
  {Uffink}(2001)}]{PhysRevA.65.012107}%
  \BibitemOpen
  \bibfield  {author} {\bibinfo {author} {\bibfnamefont {M.}~\bibnamefont
  {Seevinck}}\ and\ \bibinfo {author} {\bibfnamefont {J.}~\bibnamefont
  {Uffink}},\ }\bibfield  {title} {\bibinfo {title} {Sufficient conditions for
  three-particle entanglement and their tests in recent experiments},\ }\href
  {https://doi.org/10.1103/PhysRevA.65.012107} {\bibfield  {journal} {\bibinfo
  {journal} {Phys. Rev. A}\ }\textbf {\bibinfo {volume} {65}},\ \bibinfo
  {pages} {012107} (\bibinfo {year} {2001})}\BibitemShut {NoStop}%
\bibitem [{\citenamefont {Shimony}(1995)}]{Shimony_95}%
  \BibitemOpen
  \bibfield  {author} {\bibinfo {author} {\bibfnamefont {A.}~\bibnamefont
  {Shimony}},\ }\bibfield  {title} {\bibinfo {title} {Degree of entanglement},\
  }\href
  {https://nyaspubs.onlinelibrary.wiley.com/doi/abs/10.1111/j.1749-6632.1995.tb39008.x}
  {\bibfield  {journal} {\bibinfo  {journal} {Ann. N. Y. Acad. Sci.}\ }\textbf
  {\bibinfo {volume} {755}},\ \bibinfo {pages} {675} (\bibinfo {year}
  {1995})}\BibitemShut {NoStop}%
\bibitem [{\citenamefont {Barnum}\ and\ \citenamefont
  {Linden}(2001)}]{HBarnum_2001}%
  \BibitemOpen
  \bibfield  {author} {\bibinfo {author} {\bibfnamefont {H.}~\bibnamefont
  {Barnum}}\ and\ \bibinfo {author} {\bibfnamefont {N.}~\bibnamefont
  {Linden}},\ }\bibfield  {title} {\bibinfo {title} {Monotones and invariants
  for multi-particle quantum states},\ }\href
  {https://doi.org/10.1088/0305-4470/34/35/305} {\bibfield  {journal} {\bibinfo
   {journal} {J. Phys. A: Math. Gen.}\ }\textbf {\bibinfo {volume} {34}},\
  \bibinfo {pages} {6787} (\bibinfo {year} {2001})}\BibitemShut {NoStop}%
\bibitem [{\citenamefont {Gour}\ and\ \citenamefont
  {Wallach}(2007)}]{PhysRevA.76.042309}%
  \BibitemOpen
  \bibfield  {author} {\bibinfo {author} {\bibfnamefont {G.}~\bibnamefont
  {Gour}}\ and\ \bibinfo {author} {\bibfnamefont {N.~R.}\ \bibnamefont
  {Wallach}},\ }\bibfield  {title} {\bibinfo {title} {Entanglement of subspaces
  and error-correcting codes},\ }\href
  {https://doi.org/10.1103/PhysRevA.76.042309} {\bibfield  {journal} {\bibinfo
  {journal} {Phys. Rev. A}\ }\textbf {\bibinfo {volume} {76}},\ \bibinfo
  {pages} {042309} (\bibinfo {year} {2007})}\BibitemShut {NoStop}%
\bibitem [{\citenamefont {Branciard}\ \emph {et~al.}(2010)\citenamefont
  {Branciard}, \citenamefont {Zhu}, \citenamefont {Chen},\ and\ \citenamefont
  {Scarani}}]{PhysRevA.82.012327}%
  \BibitemOpen
  \bibfield  {author} {\bibinfo {author} {\bibfnamefont {C.}~\bibnamefont
  {Branciard}}, \bibinfo {author} {\bibfnamefont {H.}~\bibnamefont {Zhu}},
  \bibinfo {author} {\bibfnamefont {L.}~\bibnamefont {Chen}},\ and\ \bibinfo
  {author} {\bibfnamefont {V.}~\bibnamefont {Scarani}},\ }\bibfield  {title}
  {\bibinfo {title} {Evaluation of two different entanglement measures on a
  bound entangled state},\ }\href {https://doi.org/10.1103/PhysRevA.82.012327}
  {\bibfield  {journal} {\bibinfo  {journal} {Phys. Rev. A}\ }\textbf {\bibinfo
  {volume} {82}},\ \bibinfo {pages} {012327} (\bibinfo {year}
  {2010})}\BibitemShut {NoStop}%
\bibitem [{\citenamefont {Demianowicz}\ and\ \citenamefont
  {Augusiak}(2019)}]{Demianowicz_2019}%
  \BibitemOpen
  \bibfield  {author} {\bibinfo {author} {\bibfnamefont {M.}~\bibnamefont
  {Demianowicz}}\ and\ \bibinfo {author} {\bibfnamefont {R.}~\bibnamefont
  {Augusiak}},\ }\bibfield  {title} {\bibinfo {title} {Entanglement of
  genuinely entangled subspaces and states: Exact, approximate, and numerical
  results},\ }\href {https://doi.org/10.1103/PhysRevA.100.062318} {\bibfield
  {journal} {\bibinfo  {journal} {Phys. Rev. A}\ }\textbf {\bibinfo {volume}
  {100}},\ \bibinfo {pages} {062318} (\bibinfo {year} {2019})}\BibitemShut
  {NoStop}%
\bibitem [{\citenamefont {Makuta}\ \emph {et~al.}(2023)\citenamefont {Makuta},
  \citenamefont {Kuzaka},\ and\ \citenamefont
  {Augusiak}}]{Makuta2023fullynonpositive}%
  \BibitemOpen
  \bibfield  {author} {\bibinfo {author} {\bibfnamefont {O.}~\bibnamefont
  {Makuta}}, \bibinfo {author} {\bibfnamefont {B.}~\bibnamefont {Kuzaka}},\
  and\ \bibinfo {author} {\bibfnamefont {R.}~\bibnamefont {Augusiak}},\
  }\bibfield  {title} {\bibinfo {title} {Fully non-positive-partial-transpose
  genuinely entangled subspaces},\ }\href
  {https://doi.org/10.22331/q-2023-02-09-915} {\bibfield  {journal} {\bibinfo
  {journal} {{Quantum}}\ }\textbf {\bibinfo {volume} {7}},\ \bibinfo {pages}
  {915} (\bibinfo {year} {2023})}\BibitemShut {NoStop}%
\bibitem [{\citenamefont {Samoilenko}(1991)}]{978-0-7923-0703-7}%
  \BibitemOpen
  \bibfield  {author} {\bibinfo {author} {\bibfnamefont {Y.~S.}\ \bibnamefont
  {Samoilenko}},\ }\href
  {https://doi.org/https://doi.org/10.1007/978-94-011-3806-2} {\emph {\bibinfo
  {title} {Mathematics and its Applications, Vol. 57}}}\ (\bibinfo  {publisher}
  {Springer Science \& Business Media Dordrecht},\ \bibinfo {year}
  {1991})\BibitemShut {NoStop}%
\bibitem [{\citenamefont {Englbrecht}\ \emph {et~al.}(2022)\citenamefont
  {Englbrecht}, \citenamefont {Kraft},\ and\ \citenamefont
  {Kraus}}]{Englbrecht_2022}%
  \BibitemOpen
  \bibfield  {author} {\bibinfo {author} {\bibfnamefont {M.}~\bibnamefont
  {Englbrecht}}, \bibinfo {author} {\bibfnamefont {T.}~\bibnamefont {Kraft}},\
  and\ \bibinfo {author} {\bibfnamefont {B.}~\bibnamefont {Kraus}},\ }\bibfield
   {title} {\bibinfo {title} {Transformations of stabilizer states in quantum
  networks},\ }\href {https://doi.org/10.22331/q-2022-10-25-846} {\bibfield
  {journal} {\bibinfo  {journal} {Quantum}\ }\textbf {\bibinfo {volume} {6}},\
  \bibinfo {pages} {846} (\bibinfo {year} {2022})}\BibitemShut {NoStop}%
\bibitem [{\citenamefont {Fattal}\ \emph {et~al.}(2004)\citenamefont {Fattal},
  \citenamefont {Cubitt}, \citenamefont {Yamamoto}, \citenamefont {Bravyi},\
  and\ \citenamefont {Chuang}}]{fattal2004entanglementstabilizerformalism}%
  \BibitemOpen
  \bibfield  {author} {\bibinfo {author} {\bibfnamefont {D.}~\bibnamefont
  {Fattal}}, \bibinfo {author} {\bibfnamefont {T.~S.}\ \bibnamefont {Cubitt}},
  \bibinfo {author} {\bibfnamefont {Y.}~\bibnamefont {Yamamoto}}, \bibinfo
  {author} {\bibfnamefont {S.}~\bibnamefont {Bravyi}},\ and\ \bibinfo {author}
  {\bibfnamefont {I.~L.}\ \bibnamefont {Chuang}},\ }\href
  {https://arxiv.org/abs/quant-ph/0406168} {\bibinfo {title} {Entanglement in
  the stabilizer formalism}} (\bibinfo {year} {2004}),\ \Eprint
  {https://arxiv.org/abs/quant-ph/0406168} {arXiv:quant-ph/0406168 [quant-ph]}
  \BibitemShut {NoStop}%
\bibitem [{\citenamefont {Chau}(1997)}]{Chau_1997}%
  \BibitemOpen
  \bibfield  {author} {\bibinfo {author} {\bibfnamefont {H.~F.}\ \bibnamefont
  {Chau}},\ }\bibfield  {title} {\bibinfo {title} {Five quantum register error
  correction code for higher spin systems},\ }\href
  {https://doi.org/10.1103/PhysRevA.56.R1} {\bibfield  {journal} {\bibinfo
  {journal} {Phys. Rev. A}\ }\textbf {\bibinfo {volume} {56}},\ \bibinfo
  {pages} {R1} (\bibinfo {year} {1997})}\BibitemShut {NoStop}%
\bibitem [{\citenamefont {Chen}\ \emph {et~al.}(2014)\citenamefont {Chen},
  \citenamefont {Aulbach},\ and\ \citenamefont {Hajdu\ifmmode~\check{s}\else
  \v{s}\fi{}ek}}]{Chen_2014}%
  \BibitemOpen
  \bibfield  {author} {\bibinfo {author} {\bibfnamefont {L.}~\bibnamefont
  {Chen}}, \bibinfo {author} {\bibfnamefont {M.}~\bibnamefont {Aulbach}},\ and\
  \bibinfo {author} {\bibfnamefont {M.}~\bibnamefont
  {Hajdu\ifmmode~\check{s}\else \v{s}\fi{}ek}},\ }\bibfield  {title} {\bibinfo
  {title} {Comparison of different definitions of the geometric measure of
  entanglement},\ }\href {https://doi.org/10.1103/PhysRevA.89.042305}
  {\bibfield  {journal} {\bibinfo  {journal} {Phys. Rev. A}\ }\textbf {\bibinfo
  {volume} {89}},\ \bibinfo {pages} {042305} (\bibinfo {year}
  {2014})}\BibitemShut {NoStop}%
\bibitem [{\citenamefont {Contreras-Tejada}\ \emph {et~al.}(2019)\citenamefont
  {Contreras-Tejada}, \citenamefont {Palazuelos},\ and\ \citenamefont
  {de~Vicente}}]{PhysRevLett.122.120503}%
  \BibitemOpen
  \bibfield  {author} {\bibinfo {author} {\bibfnamefont {P.}~\bibnamefont
  {Contreras-Tejada}}, \bibinfo {author} {\bibfnamefont {C.}~\bibnamefont
  {Palazuelos}},\ and\ \bibinfo {author} {\bibfnamefont {J.~I.}\ \bibnamefont
  {de~Vicente}},\ }\bibfield  {title} {\bibinfo {title} {Resource theory of
  entanglement with a unique multipartite maximally entangled state},\ }\href
  {https://doi.org/10.1103/PhysRevLett.122.120503} {\bibfield  {journal}
  {\bibinfo  {journal} {Phys. Rev. Lett.}\ }\textbf {\bibinfo {volume} {122}},\
  \bibinfo {pages} {120503} (\bibinfo {year} {2019})}\BibitemShut {NoStop}%
\bibitem [{\citenamefont {Antipin}(2021)}]{Antipin_2021}%
  \BibitemOpen
  \bibfield  {author} {\bibinfo {author} {\bibfnamefont {K.~V.}\ \bibnamefont
  {Antipin}},\ }\bibfield  {title} {\bibinfo {title} {Construction of genuinely
  entangled subspaces and the associated bounds on entanglement measures for
  mixed states},\ }\href {https://doi.org/10.1088/1751-8121/ac37e5} {\bibfield
  {journal} {\bibinfo  {journal} {J. Phys. A: Math. Theor.}\ }\textbf {\bibinfo
  {volume} {54}},\ \bibinfo {pages} {505303} (\bibinfo {year}
  {2021})}\BibitemShut {NoStop}%
\bibitem [{\citenamefont {Hastings}\ and\ \citenamefont
  {O'Donnell}(2023)}]{hastings2023optimizing}%
  \BibitemOpen
  \bibfield  {author} {\bibinfo {author} {\bibfnamefont {M.~B.}\ \bibnamefont
  {Hastings}}\ and\ \bibinfo {author} {\bibfnamefont {R.}~\bibnamefont
  {O'Donnell}},\ }\href@noop {} {\bibinfo {title} {Optimizing strongly
  interacting fermionic hamiltonians}} (\bibinfo {year} {2023}),\ \Eprint
  {https://arxiv.org/abs/2110.10701} {arXiv:2110.10701 [quant-ph]} \BibitemShut
  {NoStop}%
\bibitem [{\citenamefont {Kaniewski}\ \emph {et~al.}(2019)\citenamefont
  {Kaniewski}, \citenamefont {{\v{S}}upi{\'{c}}}, \citenamefont {Tura},
  \citenamefont {Baccari}, \citenamefont {Salavrakos},\ and\ \citenamefont
  {Augusiak}}]{Kaniewski2019maximalnonlocality}%
  \BibitemOpen
  \bibfield  {author} {\bibinfo {author} {\bibfnamefont {J.}~\bibnamefont
  {Kaniewski}}, \bibinfo {author} {\bibfnamefont {I.}~\bibnamefont
  {{\v{S}}upi{\'{c}}}}, \bibinfo {author} {\bibfnamefont {J.}~\bibnamefont
  {Tura}}, \bibinfo {author} {\bibfnamefont {F.}~\bibnamefont {Baccari}},
  \bibinfo {author} {\bibfnamefont {A.}~\bibnamefont {Salavrakos}},\ and\
  \bibinfo {author} {\bibfnamefont {R.}~\bibnamefont {Augusiak}},\ }\bibfield
  {title} {\bibinfo {title} {Maximal nonlocality from maximal entanglement and
  mutually unbiased bases, and self-testing of two-qutrit quantum systems},\
  }\href {https://doi.org/10.22331/q-2019-10-24-198} {\bibfield  {journal}
  {\bibinfo  {journal} {{Quantum}}\ }\textbf {\bibinfo {volume} {3}},\ \bibinfo
  {pages} {198} (\bibinfo {year} {2019})}\BibitemShut {NoStop}%
\end{thebibliography}

\onecolumngrid
\appendix

\section{Proof of $A_{I}^{d}=\mathbb{1}$}\label{app:A_I}
We want to prove that $A_{I}$ as defined in Eq. \eqref{eq:def_A_I}, satisfies $A_{I}^{d}=\mathbb{1}$ Written explicitly, we have
\begin{equation}
A_{I}^{d}= \left(\alpha_{I}\prod_{i=1}^{k}T_{i}^{I_{i}}\right)^{d}
=\alpha_{I}^{d} T_{1}\prod_{j=2}^{k}[T_{j},T_{1}]_{\bullet}^{I_{1}I_{j}(d-1)}\left(\prod_{i=1}^{k}T_{i}^{I_{i}}\right)^{d-1}\prod_{i=2}^{k}T_{i}^{I_{i}}
=\alpha_{I}^{d}\prod_{i=1}^{k-1}\prod_{j=i+1}^{k}\left[T_{j},T_{i}\right]_{\bullet}^{I_{i}I_{j}d(d-1)/2}.
\end{equation}
Next, using the generating graph notation \eqref{eq:gamma_def}, we can simplify it as follows
\begin{equation}
A_{I}^{d}=\alpha_{I}^{d}\prod_{i=1}^{k-1}\prod_{j=i+1}^{k}\omega^{-\gamma_{i,j} I_{i}I_{j}d(d-1)/2}\mathbb{1}=\alpha_{I}^{d}\omega^{-d(d-1)/2\sum_{i=1}^{k-1}\sum_{j=i+1}^{k}\gamma_{i,j} I_{i}I_{j}}\mathbb{1}.
\end{equation}
For odd $d$, 
\begin{equation}
 \frac{d-1}{2}\sum_{i=1}^{k-1}\sum_{j=i+1}^{k}\gamma_{i,j} I_{i}I_{j}   
\end{equation}
is a natural number, and since $\omega^{d n}=1$ for all $n\in\mathbb{N}$, we have that $A_{I}^{d}=\mathbb{1}$ for $\alpha_{I}=1$. For even $d=2$, which is the only even prime number, 
\begin{equation}
 \frac{d}{2}\sum_{i=1}^{k-1}\sum_{j=i+1}^{k}\gamma_{i,j} I_{i}I_{j}   
\end{equation}
is a natural number, and since $\omega^{n}=\pm1$ for all $n\in \mathbb{N}$, we have that if
\begin{equation}
\alpha_{I}^{2}=\pm1,
\end{equation}
then $A_{I}^{2}=\mathbb{1}$. This holds true since $\alpha_{I}\in\{1,\mathbb{i}\}$, and so we can always choose appropriate $\alpha_{I}$ such that $A_{I}^{2}=\mathbb{1}$.

\section{Proof of Theorem \ref{thm:self-testing}}\label{app:theorem_1}

Here, we formulate a proof of Theorem \ref{thm:self-testing}. To this end, we first need to introduce and prove five necessary lemmas that relate group $\mathcal{A}$ and generating graph $g$. Lemma \ref{lem:self-testing} generalizes a result by Ref. \cite{PhysRevA.106.012431} allowing us to transform matrices with specific commutation relations into tensor products of generalized Pauli matrices; Lemma \ref{lem:Z(A)_cardinality} establishes a connection between $\mathcal{C}(\mathcal{A})$ and $\ker(\gamma)$; Lemma \ref{lem:O_transformation} allows one to find different generating graphs $g$ for a given frustration graph $G$; Lemma \ref{lem:rank_even} states that rank of $\gamma$ is always an even number; and Lemma \ref{lem:canonical_form} provides a canonical form of a full-rank $\gamma$.

\setcounter{lem}{0}
\begin{lem}
Let us consider a set of unitary operators $\{M_{1},M_{2},\dots,M_{2m}\}$ acting on a finite-dimensional Hilbert space $\mathcal{H}$ such that 
$M_i^d=\mathbbm{1}$ and for every pair $i\neq j$, $M_iM_j=\omega^{l}M_jM_i$ for some $l$. If its corresponding frustration graph is given by
\begin{equation}\label{eq:frustration_matrix_canonical}
\Gamma = \begin{bmatrix}
0 & -1\\
1 & 0
\end{bmatrix}\oplus \ldots\oplus\begin{bmatrix}
0 & -1\\
1 & 0
\end{bmatrix},
\end{equation}
then there exists a unitary $U: \mathcal{H} \rightarrow \bigotimes_{i=1}^{m}\mathcal{H}_{i}\otimes \mathcal{H}'$ for $\mathcal{H}_{i}=\mathbb{C}^{d}$ and some $\mathcal{H}'$ such that for all $i\in [m]$
\begin{align}\label{eq:M_i_U}
\begin{split}
U M_{2i-1}U^{\dagger}& = X_{i}\otimes\mathbb{1},\quad UM_{2i}U^{\dagger}=Z_{i}\otimes \mathbb{1},
\end{split}
\end{align}
where $X_{i}, Z_{i}$ are Pauli matrices $X, Z$ acting on $\mathcal{H}_{i}$.
\end{lem}
\begin{proof}

Let us first consider the pair $M_1$ and $M_2$. Notice, that from Eq. \eqref{eq:frustration_matrix_canonical} it follows that
\begin{equation}
    M_1M_2=\omega^{-1}M_2M_1.
\end{equation}
It was shown in Ref. \cite{Kaniewski2019maximalnonlocality} that there exists a unitary operation
$U_1:\mathcal{H}\to \mathbbm{C}^d\otimes\mathcal{H}_1'$ with $\mathcal{H}_1'$ being some Hilbert spaces of in principle unknown dimension, such that 
\begin{equation}
    U_1\, M_1\, U_1^{\dagger}=X_1\otimes\mathbbm{1},\qquad U_1\, M_2\, U_1^{\dagger}=Z_1\otimes\mathbbm{1}.
\end{equation}
Let us then notice that the fact that the remaining observables $M_i$ with $(i=3,\ldots,2m)$ 
commute with $M_1$ and $M_2$, implies that the rotated observables
$U_1\,M_i\,U_1^{\dagger}$ all must admit
\begin{equation}
    U_1\,M_i\,U_1^{\dagger}=\mathbbm{1}\otimes M_i'\qquad (i=3,\ldots,2m), 
\end{equation}
that is, they act nontrivially only on $\mathcal{H}_1'$. Moreover, since
$M_i$ are unitary observables that satisfy $M_i^d=\mathbbm{1}$, it follows that $M_i'$ are also unitary which satisfy $(M_i')^d=\mathbbm{1}$.
Let us then focus on another pair of observables $M_3$ and $M_4$. 
It follows from the frustration matrix \eqref{eq:frustration_matrix_canonical} that 
\begin{equation}
 M_3M_4=\omega^{-1}M_4M_3,
\end{equation}
and so also
\begin{equation}
 M'_{3}M'_{4}=\omega^{-1}M'_{4}M'_{3}.
\end{equation}
Then, by the same argument as with $M_{1},M_{2}$ we have
\begin{equation}
\begin{aligned}
(\mathbb{1}_{d}\otimes U_{2}) M_{3} (\mathbb{1}_{d}\otimes U_{2})^{\dagger} &= \mathbb{1}_{d}\otimes X_{2} \otimes \mathbb{1},\\
(\mathbb{1}_{d}\otimes U_{2}) M_{4} (\mathbb{1}_{d}\otimes U_{2})^{\dagger} &= \mathbb{1}_{d}\otimes Z_{2} \otimes \mathbb{1},
\end{aligned}
\end{equation}
where $U_{2}:\mathcal{H}_{1}' \rightarrow \mathbb{C}^{d}\otimes \mathcal{H}_{2}'$, and $\mathcal{H}_{2}'$ is a Hilbert space of an unknown dimension. 

Repeating this procedure $m$ times produces a unitary $U: \mathbb{H} \rightarrow (\mathbb{C}^{d})^{\otimes m}\otimes \mathcal{H}_{m}'$, where $\mathcal{H}_{m}'$ is of unknown dimension, defined as
\begin{equation}
U \coloneqq U_{1} (\mathbb{1}_{d}\otimes U_{2})(\mathbb{1}_{d}^{\otimes 2}\otimes U_{3})\ldots (\mathbb{1}_{d}^{\otimes(m-1)}\otimes U_{m}).
\end{equation}
Such $U$ transforms $M_{i}$ as in Eq. \eqref{eq:M_i_U} which ends the proof.
\end{proof}

\begin{lem}\label{lem:Z(A)_cardinality}
Let $\mathcal{A}$ be a group as in Definition \ref{def: A}. Then, for every generating graph $g$ we have
\begin{equation}
|\mathcal{C}(\mathcal{A})|=d^{\operatorname{null}(\gamma)},
\end{equation}
where $\operatorname{null}$ is the nullity, i.e., the dimension of the kernel of $\gamma$.
\end{lem}
\begin{proof}
To prove the above, we show that there exists a bijection between elements of $\mathcal{C}(\mathcal{A})$ and elements of $\ker(\gamma)$. We start by showing that with every element in $\mathcal{C}(\mathcal{A})$ we can associate an element in $\ker(\gamma)$.

The definition of $\mathcal{C}(\mathcal{A})$ together with the definition of the frustration graph \eqref{eq:Gamma_def} implies that if $A_{I}\in \mathcal{C}(\mathcal{A})$ then for all $J\in \mathbb{Z}_{d}^{k}$ we have
\begin{equation}
\Gamma_{I,J}=0.
\end{equation}
Using Eq. \eqref{eq:Gamma_gamma} we can rewrite it in terms of $\gamma$:
\begin{equation}
\sum_{i=1}^{k}\sum_{j=1}^{k} I_{i}J_{j}\gamma_{i,j}=0
\end{equation}
for all $J\in \mathbb{Z}_{d}^{k}$ and so in particular
\begin{equation}
\sum_{i=1}^{k}I_{i}\gamma_{i,j}=0
\end{equation}
for all $j\in[k]$. Let $e_{i}$ be a $k$-dimensional unit vector, i.e., it has a $1$ entry on the $i$'th position and $0$ elsewhere. Then we have
\begin{equation}
\gamma I=\sum_{i,j=1}^{k}\gamma_{j,i}e_{j}e_{i}^{T} I_{i}e_{i}=\sum_{j=1}^{k}\sum_{i=1}^{k}I_{i}\gamma_{j,i}e_{j}=-\sum_{j=1}^{k}\sum_{i=1}^{k}I_{i}\gamma_{i,j}e_{j}=-\sum_{j=1}^{k}0 e_{j}=0,
\end{equation}
where we used the fact $\gamma_{j,i}=-\gamma_{i,j} \mod d$. And so, for every operator $A_{I}\in \mathcal{C}(\mathcal{A})$ we have $I\in\ker (\gamma)$.

To show that the converse is also true, i.e., that if $I\in\ker (\gamma)$ then $A_{I}\in \mathcal{C}(\mathcal{A})$, we start from 
\begin{align}
\begin{split}
\gamma I&=0,\\
\sum_{i=1}^{k}\gamma_{j,i}I_{i}&=0\quad \textrm{for all }j\in [k],\\
\sum_{j=1}^{k}\sum_{i=1}^{k}I_{i}J_{j}\gamma_{j,i}&=0\quad \textrm{for all }J\in\mathbb{Z}_{d}^{k},\\
\Gamma_{J,I}&=0\quad \textrm{for all }J\in\mathbb{Z}_{d}^{k}.
\end{split}
\end{align}
The last equation implies $A_{I}\in \mathcal{C}(\mathcal{A})$ since $\Gamma_{I,J}=-\Gamma_{J,I} \mod d$.

Therefore, we have shown that there exists a one-to-one association between elements of $\ker(\gamma)$ and elements of $\mathcal{C}(\mathcal{A})$. Since the number of distinct vectors in $\ker(\gamma)$ equals $d^{\operatorname{null}(\gamma)}$ it follows that
\begin{equation}
|\mathcal{C}(\mathcal{A})| = d^{\operatorname{null}(\gamma)}.
\end{equation}
\end{proof}

\begin{lem}\label{lem:O_transformation}
Let $\mathcal{A}$ be a group as in Definition \ref{def: A}, and let $\gamma$ correspond to a generating set $\{T_{i}\}_{i=1}^{k}$ of $\mathcal{A}$. Then there exists a one-to-one correspondence between transformations of a generating set of $\mathcal{A}$, $\phi:\{T_{i}\}_{i=1}^{k}\rightarrow \{T_{i}'\}_{i=1}^{k}$ and a transformation of the corresponding generating graphs $O: \gamma\rightarrow\gamma'$ given by
\begin{equation}\label{eq:gamma_transformation}
\gamma' = O^{T}\gamma O.
\end{equation}
\end{lem}

\begin{proof}
Let us denote the generating set of $\mathcal{A}$ corresponding to $\gamma$ as $\{T_{i}\}_{i=1}^{k}$ and let $\phi:\{T_{i}\}_{i=1}^{k}\rightarrow \{T_{i}'\}_{i=1}^{k}$ be defined as
\begin{equation}\label{eq:phi_def}
T_{i}'=\prod_{j=1}^{k} T_{j}^{O_{j,i}}.
\end{equation}
This equation establishes a one-to-one correspondence between $\phi$ and $O$ by virtue of Ref. \cite[Theorem 6.3]{Sarkar_2024}. Therefore, to prove the lemma, all we have to do is to show that such a transformation satisfies \eqref{eq:gamma_transformation}, or in other words, that $\gamma'$ is the adjacency matrix of the generating graph of $\{T_{i}'\}_{i=1}^{k}$.
There are two important consequences of Eq. \eqref{eq:phi_def}. First, the fact that the generators in both generating sets have to be independent implies that $O$ is invertible. Second, it allows us to derive the following relation
\begin{equation}
[T_{i}',T_{j}']_{\bullet} = \prod_{r,s= 1}^{k} [T_{r},T_{s}]_{\bullet}^{O_{r,i} O_{s,j}}.
\end{equation}
Then we can use the relation
\begin{equation}
[T_{r},T_{s}]_{\bullet} = \omega^{\gamma_{r,s}}\mathbb{1}
\end{equation}
to infer
\begin{equation}
[T_{i}',T_{j}']_{\bullet} = \prod_{r,s=1}^{k} \omega^{O_{r,i}\gamma_{r,s}O_{s,j}}=\omega^{\gamma_{i,j}'},
\end{equation}
which proves that the commutation relations of $\{T_{i}'\}_{i=1}^{k}$ are described by $\gamma'$.
\end{proof}

\begin{lem}\label{lem:rank_even}
Let $\mathcal{A}=\langle T_{1},\dots,T_{k}\rangle_{\odot}$ be a group as in Definition \ref{def: A}. For all such $\mathcal{A}$, the rank of $\gamma$ is even.
\end{lem}
While this lemma follows as a consequence of Ref. \cite[Lemma 5.3]{Sarkar_2024}, we nonetheless prove it here, as the substance of the proof will also be used for a different proof.

\begin{proof}
We first need to show that for each adjacency matrix $\gamma\in M_{k\cross k}(\mathbb{Z}_{d})$ we can find an invertible transformation $O\in M_{k\cross k}(\mathbb{Z}_{d})$ such that
\begin{equation}\label{eq:gamma2}
O^{T}\gamma O= \begin{bmatrix}
0 & D \\
-D^{T} & E
\end{bmatrix},
\end{equation}
where $D\in M_{n\times m}(\mathbb{Z}_{d})$ is full rank. To this end, let us observe that for all $\gamma$ we have
\begin{equation}
\gamma = \begin{bmatrix}
0 & D_{0} \\
-D_{0}^{T} & E_{0}
\end{bmatrix},
\end{equation}
where $D_{0}\in M_{n_{0}\cross m_{0}}(\mathbb{Z}_{d})$ (in the worst-case scenario $n_{0}=1$). If $D_{0}$ is not full-rank, then we identify one of its columns that is linearly dependent on the rest, and we find an invertible operation $O_{1}\in M_{k\cross k}(\mathbb{Z}_{d})$ such that it transforms the aforementioned column into a column of $0$'s, and afterward permutes this column such that it becomes the first column from the left of $D_{0}$. Notice that the top left element of $E_{0}$ is always $0$, which allows us to increase the $0$ block of the matrix by one row and column. Therefore, we have
\begin{equation}
    O_{1}^{T}\gamma O_{1} = \begin{bmatrix}
0 & D_{1} \\
-D_{1}^{T} & E_{1}
\end{bmatrix},
\end{equation}
where $D_{1}\in M_{n_{1}\cross m_{1}}(\mathbb{Z}_{d})$, and $n_{1}=n_{0}+1$, $m_{1}=m_{0}-1$. We can now repeat this procedure, up until we get $D_{l}=D$ that is full-rank. This finally yields Eq. \eqref{eq:gamma2} with $O=O_{1}O_{2}\ldots O_{l}$.

Next, notice that $\rank(-D^{T})=\rank(D)=m$ implies $\operatorname{null}(-D^{T})= n-m$.  Since every vector from $\ker(-D^{T})$ can be mapped to a vector from $\ker{\gamma}$, and since $D$  is full rank, we conclude that
\begin{equation}
\operatorname{null}(\gamma)=n-m.
\end{equation}
Finally, substituting $n=k-m$ and $k=\operatorname{null}(\gamma)+\rank(\gamma)$ to the above yields
\begin{equation}\label{eq:rank_even}
\rank (\gamma) = 2m.
\end{equation}
\end{proof}
\begin{cor}\label{cor:k_even}
If $\operatorname{null}(\gamma)=0$, then $k$ is even.
\end{cor}
\begin{proof}
If $\operatorname{null}(\gamma)=0$ then $\rank(\gamma)=k$, and so by Lemma \ref{lem:rank_even} $k$ is even.
\end{proof}

Lastly, let us state the \cite[Lemma 5.3]{Sarkar_2024} as the following lemma.
\begin{lem}\label{lem:canonical_form}
Let $\gamma$ be a full-rank adjacency matrix of a generating graph. For every such $\gamma$ there exist an invertible operation $O\in M_{k\cross k}(\mathbb{Z}_{d})$ such that
\begin{equation}
O^{T}\gamma O= \begin{bmatrix}
0 & -1\\
1 & 0
\end{bmatrix}\oplus\begin{bmatrix}
0 & -1\\
1 & 0
\end{bmatrix}\oplus\ldots\oplus \begin{bmatrix}
0 & -1\\
1 & 0
\end{bmatrix}.
\end{equation}
\end{lem}

With all of the lemmas proven, we are finally ready to prove Theorem \ref{thm:self-testing}.

\setcounter{thm}{0}
\begin{thm}
Let $\mathcal{A}=\langle T_{1},\dots, T_{k}\rangle_{\odot}$ be a group as in Definition \ref{def: A} and let $\gamma$ be a generating graph corresponding to $\{T_{i}\}_{i=1}^{k}$. There exists a unitary $U$ such that for every element from $\mathcal{A}$ one has
\begin{equation}
U A\, U^{\dagger}=P_{A}\otimes C_{A},
\end{equation}
where $C_{A}$ is a unitary matrix such that $C_{A}^{d}=\mathbb{1}$ and $[C_{A},C_{A'}]=0$ for all $A,A'\in\mathcal{A}$, and $P_{A}\in \tilde{\mathbb{P}}_{q/2}$ for $q=\operatorname{rank}(\gamma)$.
\end{thm}
\begin{proof}
By the virtue of Lemma \ref{lem:Z(A)_cardinality} we can choose a generating set $\{T_{i}\}_{i=1}^{k}$ such that $\mathcal{S}(\mathcal{A})=\langle T_{1},\ldots, T_{q}\rangle_{\odot}$. Let $\tilde{\gamma}$ be the adjacency matrix of the generating graph $\tilde{g}$ corresponding to the generating set of $\mathcal{S}(\mathcal{A})$. By the virtue of Lemma \ref{lem:Z(A)_cardinality} we have that $|\mathcal{C}(\mathcal{S}(\mathcal{A}))|=1$ implies $\operatorname{null}(\tilde{\gamma})=0$, i.e., $\tilde{\gamma}$ is full rank. Then, from Corollary \ref{cor:k_even}, it follows that $q$ is even.

From Lemma \ref{lem:canonical_form} it follows that for any $\gamma$ that is full rank there exists an invertible operation $O\in M_{q\cross q} (\mathbb{Z}_{d})$ 
\begin{equation}
O^{T}\gamma O =\begin{bmatrix}
0 & -1\\
1 & 0
\end{bmatrix}\oplus\begin{bmatrix}
0 & -1\\
1 & 0
\end{bmatrix}\oplus\ldots\oplus \begin{bmatrix}
0 & -1\\
1 & 0
\end{bmatrix}.
\end{equation}

By the virtue of Lemma \ref{lem:O_transformation}, there exists a generating set $\{T_{i}'\}_{i=1}^{q}$ of $\mathcal{S}(\mathcal{A})$ that corresponds to $\gamma'=O^{T}\tilde{\gamma}O$. The existence of such a generating set allows us to directly use Lemma \ref{lem:self-testing}: there exists a unitary $U$ such that
\begin{align}
\begin{split}
U T_{2i-1}'U^{\dagger} = X_{i}\otimes \mathbb{1},\qquad UT_{2i}'U^{\dagger}=Z_{i}\otimes\mathbb{1}
\end{split}
\end{align}
for all $i\in [q/2]$. Then, it follows from Eq. \eqref{eq:def_A_I} that every $A\in\mathcal{S}(\mathcal{A})$ we have that
\begin{equation}
U A U^{\dagger}= P_{A} \otimes \mathbb{1}
\end{equation}
for some $P_{A}\in \tilde{\mathbb{P}}_{q/2}$.

Let us now consider the subgroup $\mathcal{C}(\mathcal{A})$. Since every $A_{C}\in \mathcal{C}(\mathcal{A})$ commutes with every element from the subgroup $\langle T_{1},\dots,T_{q}\rangle$ we have that
\begin{equation}\label{eq:center_self-testing}
UA_{C}U^{\dagger} = \mathbb{1}_{q/2}\otimes C_{A_{C}},
\end{equation}
where $C_{A_{C}}$ is some unitary matrix such that $C_{A_{C}}^{d}=\mathbb{1}$, and $\mathbb{1}_{q/2}$ acts on $\bigotimes_{i=1}^{q/2}\mathcal{H}_{i}=(\mathbb{C}^{d})^{\otimes q/2}$. From the fact that every pair $A_{C},A_{C}'\in \mathcal{C}(\mathcal{A})$ commutes, we have that $[C_{A_{C}},C_{A_{C}}']_{\bullet}=\mathbb{1}$ for all $C_{A_{C}},C_{A_{C}}'$. 

Lastly, from the fact that for every element $A\in\mathcal{A}$ there exist $A_{S}\in \mathcal{S}(\mathcal{A})$ and $A_{C} \in \mathcal{C}(\mathcal{A})$ such that $A=A_{S}\odot A_{C}$, we can infer
\begin{equation}
U A U^{\dagger} = U A_{S} U^{\dagger} U A_{C} U^{\dagger} = P_{A_{S}}\otimes C_{A_{C}},
\end{equation}
which ends the proof.
\end{proof}
\section{Proof of Theorem \ref{thm:sos}}\label{app:sos}

In this section, we give the full proof of Theorem \ref{thm:sos}. To start, let us introduce a fact and a lemma that will be crucial for the proof of the upper bound.

\begin{fact}\label{fact:swap}
The unitary $U_{\textrm{SWAP}}$ that swaps the order of two qudits equals
\begin{equation}\label{SWAPd0}
U_{\textrm{SWAP}}=\frac{1}{d}\sum_{i,j=0}^{d-1}[X^iZ^j\otimes (X^iZ^j)^{\dagger}].
\end{equation}
\end{fact}
\begin{proof}
Operator $U_{\textrm{SWAP}}$ is defined by the relation
\begin{equation}
\forall_{a,b\in \mathbb{Z}_{d}} \qquad U_{\textrm{SWAP}} \ket{a}\!\ket{b}=\ket{b}\!\ket{a}.
\end{equation}
It is easy to see that
\begin{equation}
\begin{aligned}
&\frac{1}{d}\sum_{i,j=0}^{d-1}[X^iZ^j\otimes (X^iZ^j)^{\dagger}]\ket{a}\!\ket{b}=\frac{1}{d}\sum_{i,j=0}^{d-1}\omega^{j(a-b+i)}\ket{a+i}\!\ket{b-i}=\ket{b}\!\ket{a}
\end{aligned}
\end{equation}
\end{proof}

\begin{lem}\label{lem:clique_number}
Let $\mathcal{A}=\langle T_{1},\dots,T_{k}\rangle_{\odot}$ be a group as in Definition \ref{def: A} and let $\overline{G}$ be the corresponding commutation graph. Then for every $\mathcal{A}$ we have
\begin{equation}\label{eq:commutation_clique_bound}
\tilde{\omega}(\overline{G})=d^{(\operatorname{null}(\gamma)+k)/2}.
\end{equation}
\end{lem}
\begin{proof}
Let $g$ be a generating graph of $G$ and let $\overline{g}$ be a graph with adjacency matrix $\overline{\gamma}$ defined as follows
\begin{equation}
\overline{\gamma}_{i,j}=\delta_{\gamma_{i,j},0},
\end{equation}
where $\delta_{\gamma_{i,j},0}$ is the Kronecker delta. Then, using the same argument as in the proof of Lemma \ref{lem:rank_even}, one can show that there exists an invertible transformation $O\in M_{k\cross k}(\mathbb{Z}_{d})$ such that
\begin{equation}\label{eq:gamma3}
O^{T}\gamma O= \begin{bmatrix}
0 & D \\
-D^{T} & E
\end{bmatrix},
\end{equation}
where $D\in M_{n\times m}(\mathbb{Z}_{d})$ is full rank, which then implies that
\begin{equation}\label{eq:null}
\operatorname{null}(\gamma) = n-m.
\end{equation}
Notice, that the $0$ block in \eqref{eq:gamma3} for $g$, represents a clique in $\overline{g}$. Since $D$ is full rank, the size of the $0$ block cannot be increased and so $n= \max_{\overline{g}} \tilde{\omega}(\overline{g})$. Substituting $m=k-n$ and $n= \max_{\overline{g}} \tilde{\omega}(\overline{g})$ we get
\begin{equation}\label{eq:dimker_g}
\operatorname{null}(\gamma)= 2 \max_{\overline{g}} \tilde{\omega}(\overline{g}) - k.
\end{equation}

Next, let us examine how $\tilde{\omega}(\overline{G})$ can be expressed as a function of $\tilde{\omega}(\overline{g})$. First, notice that the largest clique of a commutation graph $\overline{G}$ corresponds to a subgroup of $\mathcal{A}$ in which all elements mutually commute with respect to the matrix multiplication, since given two operators $A_{1}$ and $A_{2}$ in this clique, clearly operator $A_{1}\odot A_{2}$ is also in the clique. Let $\{T_{i}'\}_{i=1}^{q}$ be the generating set of this subgroup. Then for the corresponding graph $\overline{g}'$, we have $\tilde{\omega}(\overline{G}) = d^{\tilde{\omega}(\overline{g}')}$. Moreover, for any clique in any graph $\overline{g}$, operators from this clique generate a clique in $\overline{G}$. Therefore for all $\overline{g}$ we have $\tilde{\omega}(\overline{G})\geqslant d^{\tilde{\omega}(\overline{g})}$. Putting these two fact together gives us 
\begin{equation}\label{eq:omega_G}
\tilde{\omega}(\overline{G})=\max_{\overline{g}} d^{\tilde{\omega}(\overline{g})}.
\end{equation}
Finally, by the virtue of Eq. \eqref{eq:dimker_g} and Eq. \eqref{eq:omega_G} we have that
\begin{equation}\label{clique_number_com}
\tilde{\omega}(\overline{G}) = d^{(\operatorname{null}(\gamma) + k)/2}.
\end{equation}
\end{proof}

We can now proceed to the proof of Theorem \ref{thm:sos}
\setcounter{thm}{1}
\begin{thm}
Let $\mathcal{A}=\langle T_{1},\dots,T_{k}\rangle_{\odot}$ be a group as in Definition \ref{def: A}. For each such $\mathcal{A}$ we have
\begin{equation}\label{eq:sos_bound}
\sum_{A\in\mathcal{A}}|\langle A\rangle|^{2}\leqslant d^{(\operatorname{null}(\gamma)+k)/2}=\tilde{\omega}(\overline{G}).
\end{equation}
\end{thm}
\begin{proof}
By virtue of Theorem \ref{thm:self-testing} we can rewrite the sum on the left-hand side of Ineq. \eqref{eq:sos_bound} as
\begin{equation}
\sum_{A\in\mathcal{A}}|\langle A\rangle|^{2}=\sum_{A\in\mathcal{A}}|\langle U^{\dagger}P_{A}\otimes C_{A} U\rangle|^{2}.
\end{equation}
Trivially, we can bound this expression by taking the maximum over all states $\rho$
\begin{equation}
\sum_{A\in\mathcal{A}}|\langle U^{\dagger}P_{A}\otimes C_{A}U\rangle|^{2} \leqslant \max_{\rho} \sum_{A\in\mathcal{A}}|\operatorname{Tr}[ (P_{A}\otimes C_{A}) \rho]|^{2},
\end{equation}
where $U$ has been absorbed into the state. We know from Eq. \eqref{eq:center_self-testing} that $U A_{C} U^{\dagger}=\mathbb{1}\otimes C_{A_{C}}$. Let us choose a generating set $\{T_{i}\}_{i=1}^{k}$ of $\mathcal{A}$ such that $\mathcal{C}(\mathcal{A}) = \langle T_{q+1},\ldots, T_{k}\rangle_{\odot}$ for some $q$, and let $\mathcal{S}(\mathcal{A})=\langle T_{1},\ldots,T_{q}\rangle_{\odot}$. This allows us to rewrite the above sum as
\begin{align}
\begin{split}
&\max_{\rho} \sum_{A\in\mathcal{A}}|\operatorname{Tr}[ (P_{A}\otimes C_{A}) \rho]|^{2}=\max_{\rho} \sum_{A_{C}\in \mathcal{C}(\mathcal{A})}\sum_{A_{S}\in \mathcal{S}(\mathcal{A})}|\operatorname{Tr}[ (P_{A_{S}}\otimes C_{A_{C}}) \rho ]|^{2}
\end{split}
\end{align}

Let us focus on a term
\begin{equation}
\sum_{A_{S}\in\mathcal{S}(\mathcal{A})}|\operatorname{Tr}[(P_{A_{S}}\otimes C_{A_{C}})\rho]|^{2}
\end{equation}
for some $A_{C}\in \mathcal{C}(\mathcal{A})$. Following Theorem \ref{thm:self-testing}, $P_{A_{S}}\in \tilde{\mathbb{P}}_{q/2}$. Since $P_{A_{S}}$ is unique for every $A_{S}\in\mathcal{S}(\mathcal{A})$, and since $|\mathcal{S}(\mathcal{A})|=|\tilde{\mathbb{P}}_{q/2}|=d^{q}$ we can rewrite the sum over $A_{S}$ as
\begin{equation}\label{eq:sos_S_q}
\begin{aligned}
\sum_{A_{S}\in\mathcal{S}(\mathcal{A})}|\operatorname{Tr}[(P_{A_{S}}\otimes C_{A_{C}})\rho]|^{2}&= \sum_{P\in\tilde{\mathbb{P}}_{q/2}}|\operatorname{Tr}[(P\otimes C_{A_{C}})\rho]|^{2}= \operatorname{Tr}[(S_{q}\otimes C_{A_{C}}\otimes C_{A_{C}}^{\dagger})\rho \otimes \rho]
\end{aligned}
\end{equation}
where in the second equality we used the relationship between trace and tensor product, and
\begin{equation}
\begin{aligned}
S_{q} &= \sum_{P\in\tilde{\mathbb{P}}_{q/2}} P\otimes P^{\dagger}
\end{aligned}
\end{equation}
Note, that by definition of $\tilde{\mathbb{P}}_{q/2}$, each operator $P$ is given by o
\begin{equation}
P=\mu_{\mathbf{i},\mathbf{j}} W_{\mathbf{i},\mathbf{j}}=\mu_{\mathbf{i},\mathbf{j}} X^{i_{1}}Z^{j_{1}}\otimes X^{i_{2}}Z^{j_{2}}\otimes\ldots\otimes X^{i_{q/2}}Z^{j_{q/2}}.
\end{equation}
Using this notation we can write $S_{q}$ as
\begin{align}
\begin{split}
S_{q}= \sum_{\mathbf{i},\mathbf{j}\in \mathbb{Z}_{d}^{q/2}} \mu_{\mathbf{i},\mathbf{j}} W_{\mathbf{i},\mathbf{j}}\otimes (\mu_{\mathbf{i},\mathbf{j}} W_{\mathbf{i},\mathbf{j}})^{\dagger}= \bigotimes_{l=1}^{q/2}\sum_{i_{l},j_{l}=0}^{d-1} X^{i_{l}}Z^{j_{l}}\otimes (X^{i_{l}}Z^{j_{l}})^{\dagger}= \left(\sum_{i,j=0}^{d-1} X^{i}Z^{j}\otimes (X^{i}Z^{j})^{\dagger}\right)^{\otimes  q/2}.
\end{split}
\end{align}
Clearly, from Fact \ref{fact:swap} it then follows that
\begin{equation}
S_{q}= d^{q/2} U_{\textrm{SWAP}}^{\otimes (q/2)}.
\end{equation}
We substitute it into Eq. \eqref{eq:sos_S_q} which yields
\begin{equation}\label{eq:U_SWAP_C}
d^{q/2}\operatorname{Tr}[(U_{\textrm{SWAP}}^{\otimes (q/2)}\otimes C_{A_{C}}\otimes C_{A_{C}}^{\dagger})\rho \otimes \rho].
\end{equation}

Let us consider an arbitrary pure state $\ket{\psi}=\ket{\psi_{q/2}}\otimes \ket{\psi_{C}}$ for $\ket{\psi_{q/2}}\in \mathcal{H}_{q/2}$ and $\ket{\psi_{C}}\in\mathcal{H}_{C}$, where the Hilbert spaces $\mathcal{H}_{q/2}$ and $\mathcal{H}_{C}$ are such that $U_{\textrm{SWAP}}^{\otimes q/2}\in \mathcal{B}(\mathcal{H}_{q/2}^{\otimes 2})$ and $C_{A_{C}}\in \mathcal{B}(\mathcal{H}_{C})$. Then for any such $\ket{\psi}$ we have
\begin{equation}
\begin{aligned}
&\bra{\psi}\!\bra{\psi} U_{\textrm{SWAP}}^{\otimes q/2} \otimes C_{A_{C}} \otimes C_{A_{C}}^{\dagger} \ket{\psi}\!\ket{\psi}= \bra{\psi_{q/2}}\!\bra{\psi_{q/2}} U_{\textrm{SWAP}}^{\otimes q/2}\ket{\psi_{q/2}}\!\ket{\psi_{q/2}} |\bra{\psi_{C}} C_{A_{C}}\ket{\psi_{C}}|^{2}= |\bra{\psi_{C}} C_{A_{C}}\ket{\psi_{C}}|^{2}\leqslant 1.
\end{aligned}
\end{equation}
The above implies that Eq. \eqref{eq:U_SWAP_C} can be upper bounded by
\begin{equation}
d^{q/2}\operatorname{Tr}[(U_{\textrm{SWAP}}^{\otimes (q/2)}\otimes C_{A_{C}}\otimes C_{A_{C}}^{\dagger})\rho \otimes \rho]\leqslant d^{q/2}.
\end{equation}
Importantly, this holds true for any $A_{C}\in \mathcal{C}(\mathcal{A})$, which gives us
\begin{equation}
\sum_{A\in \mathcal{A}}|\langle A \rangle|^{2} \leqslant |\mathcal{C}(\mathcal{A})| d^{q/2}.
\end{equation}
Since by definition there are $k-q$ generators of $\mathcal{C}(\mathcal{A})$, we have that $|\mathcal{C}(\mathcal{A})|=d^{k-q}$. However, from Lemma \ref{lem:Z(A)_cardinality} we have $|\mathcal{C}(\mathcal{A})|=d^{\operatorname{null}(\gamma)}$, hence $q=k-\operatorname{null}(\gamma)$ and so
\begin{equation}
\sum_{A\in \mathcal{A}}|\langle A \rangle|^{2} \leqslant d^{(\operatorname{null}(\gamma)+k)/2},
\end{equation}
which ends the proof.
\end{proof}

\section{Proof of Theorem \ref{thm:GM}}\label{app:entanglement_measure}

\begin{thm}
Let $\mathbb{S}=\expect{g_{1},\dots,g_{k}}$ be a stabilizer with a corresponding stabilizer subspace $\mathcal{V}_{\mathbb{S}}$. The minimal geometric measure of entanglement of the subspace $\mathcal{V}_{\mathbb{S}}$ with respect to the bipartition $Q|\overline{Q}$ is given by
\begin{equation}
E_{\textrm{GM}}^{Q}(\mathcal{V}_{\mathbb{S}})=1-d^{-\rank(\gamma_{Q})/2}=1-d^{-k}\tilde{\omega}(\overline{G}_{Q}),
\end{equation}
where $\gamma_{Q}$ is an adjacency matrix of a generating graph corresponding to $\{g_{i}^{(Q)}\}_{i=1}^{k}$, and $\overline{G}_{Q}$ is the commutation graph of $\{s^{(Q)}\}_{s\in\mathbb{S}}$.
\end{thm}
\begin{proof}
Let $\mathcal{V}_{\mathbb{S}}$ be a stabilizer subspace corresponding to a stabilizer $\mathbb{S}=\expect{g_{1},\dots,g_{k}}$. The projector $\mathcal{P}_{\mathcal{V}_{\mathbb{S}}}$ onto $\mathcal{V}_{\mathbb{S}}$ is given by
\begin{equation}\label{Projector}
    \mathcal{P}_{\mathcal{V}_{\mathbb{S}}}=\frac{1}{d^{k}}\sum_{s\in\mathbb{S}}s.
\end{equation}
By substituting Eq. \eqref{Projector} into Eq. \eqref{eq:GGM final form} we arrive at
\begin{align}\label{E_GGM_general1}
\begin{split}
    E_{\textrm{GM}}^{Q}(\mathcal{V}_{\mathbb{S}})&=1-\frac{1}{d^{k}}\max_{\ket{\psi}\in\Phi_{Q}}\sum_{s\in \mathbb{S}}\Tr\left(s\ket{\psi}\bra{\psi}\right),
\end{split}    
\end{align}
where, as a reminder, $\Phi_{Q}$ is a set of all states that are product with respect to the bipartition $Q|\overline{Q}$. Let us write the bipartite state explicitly $\ket{\psi}=\ket{\phi}_{Q}\ket{\chi}_{\overline{Q}}$, where $\ket{\phi}_{Q}\in(\mathbb{C}^{2})^{\otimes|Q|}$ and $\ket{\chi}_{\overline{Q}}\in (\mathbb{C}^{2})^{\otimes(N-|Q|)}$. In the same manner, we can also write explicitly the bipartition of each $s\in\mathbb{S}$ as $s=s^{(Q)}\otimes s^{(\overline{Q})}$. Without a loss of generality, we also require that $(s^{(Q)})^{d}=(s^{(\overline{Q})})^{d}=\mathbb{1}$.

Expressing $E_{\textrm{GM}}^{Q}(\mathcal{V}_{\mathbb{S}})$ in terms of explicit forms of $\ket{\phi}$ and $s$ results in
\begin{align}\label{E_GGM_before_real}
\begin{split}
E_{\textrm{GM}}^{Q}(\mathcal{V}_{\mathbb{S}})&= 1 - \frac{1}{d^{k}}\max_{\ket{\phi}_{Q}, \ket{\chi}_{\overline{Q}}} \sum_{s\in\mathbb{S}}\Tr[s\ket{\phi}_{Q}\!\bra{\phi}\otimes\ket{\chi}_{\overline{Q}}\!\bra{\chi}]\geqslant 1 - \frac{1}{d^{k}}\max_{\ket{\phi}_{Q}, \ket{\chi}_{\overline{Q}}}\left|\sum_{s\in\mathbb{S}}\expect{s^{(Q)}}_{\ket{\phi}_{Q}}\expect{s^{(\overline{Q})}}_{\ket{\chi}_{\overline{Q}}}\right|.
\end{split}
\end{align}
Applying Cauchy-Schwarz inequality to the above gives us
\begin{equation}\label{E_GGM_final_form}
    E_{\textrm{GM}}^{Q}(\mathcal{V}_{\mathbb{S}})\geqslant 1 - \frac{1}{d^{k}}\sqrt{M_{Q}}\sqrt{M_{\overline{Q}}}.
\end{equation} 
where
\begin{equation}\label{eq:mq}
M_{Q}=\max_{\ket{\phi}_{Q}}\sum_{s\in\mathbb{S}}\left|\expect{s^{(Q)}}_{\ket{\phi}_{Q}}\right|^{2}
\end{equation}
and analogously for $M_{\overline{Q}}$. Crucially, the set $\{s^{(Q)}\}_{s\in\mathbb{S}}$ can be associated with a group $\mathcal{A}=\langle T_{1},\ldots,T_{k}\rangle_{\odot}$, and so it follows from Theorem \ref{thm:sos} that
\begin{equation}
M_{Q} \leqslant d^{(\operatorname{null}(\gamma_{Q})+k)/2},
\end{equation}
where $\gamma_{Q}$ is a generating graph corresponding to the set $\{g_{i}^{(Q)}\}_{i=1}^{k}$. It is easy to see that, due to the mutual commutation of the generators $g_{i}$, $\gamma_{\overline{Q}}$ satisfies $\gamma_{\overline{Q}}=-\gamma_{Q} \mod d$ and so $\operatorname{null}(\gamma_{\overline{Q}})=\operatorname{null}(\gamma_{Q})$. Therefore 
\begin{equation}
M_{\overline{Q}} \leqslant d^{(\operatorname{null}(\gamma_{Q})+k)/2}.
\end{equation}
We can use these relations and the fact that $k=\operatorname{null}(\gamma_{Q})+\rank(\gamma_{Q})$ to formulate the following lower bound:
\begin{equation}
\begin{aligned}
E_{\textrm{GM}}^{Q}(\mathcal{V}_{\mathbb{S}})\geqslant1 - \frac{1}{d^{k}}\sqrt{M_{Q}}\sqrt{M_{\overline{Q}}} &\geqslant  1 - d^{-\rank(\gamma_{Q})/2}= 1- d^{-k}\tilde{\omega}(\overline{G}_{Q}).
\end{aligned}
\end{equation}
where the equality follows from Eq. \eqref{clique_number_com}.

Let us now show that the above bound can always be saturated. To this end, we come back to the sum in \eqref{E_GGM_before_real}
\begin{equation}
\sum_{s\in\mathbb{S}}\expect{s^{(Q)}}_{\ket{\phi}_{Q}}\expect{ s^{(\overline{Q})}}_{\ket{\chi}_{\overline{Q}}},
\end{equation}
where now $\ket{\phi}_{Q}$ and $\ket{\chi}_{\overline{Q}}$ are some arbitrary states. Using Theorem \ref{thm:self-testing} we can transform $s^{(Q)}$ and $s^{(\overline{Q})}$ into
\begin{equation}
s^{(Q)}=U_{Q}P_{s^{(Q)}}\otimes C_{s^{(Q)}}U_{Q}^{\dagger},\quad s^{(\overline{Q})}=U_{\overline{Q}}P_{s^{(\overline{Q})}}\otimes C_{s^{(\overline{Q})}}U_{\overline{Q}}^{\dagger}.
\end{equation}
which gives us
\begin{equation}\label{eq:sum_s_total}
\sum_{s\in\mathbb{S}}\expect{P_{s^{(Q)}}\otimes C_{s^{(Q)}}\otimes P_{s^{(\overline{Q})}}\otimes C_{s^{(\overline{Q})}}}_{\ket{\phi'}_{Q}\ket{\chi'}_{\overline{Q}}},
\end{equation}
where $\ket{\phi'}_{Q}=U_{Q}^{\dagger}\ket{\phi}_{Q}$ and $\ket{\chi'}_{\overline{Q}}=U_{\overline{Q}}^{\dagger}\ket{\chi}_{\overline{Q}}$.

Let us consider the largest clique of $\overline{G}_{Q}$ and the operators $s^{(Q)}$ that correspond to the vertices of this clique. Notice, that from the mutual commutation of the elements of $\mathbb{S}$ it follows that for any bipartition $Q|\overline{Q}$ the graphs $\overline{G}_{Q}$ and $\overline{G}_{\overline{Q}}$ are equivalent. Then it follows that the same subset of $s$ corresponds to operators $s^{(Q)}$ from the largest clique of $\overline{G}_{Q}$ and to operators $s^{(\overline{Q})}$ from the largest clique of $\overline{G}_{\overline{Q}}$. We denote this subset of $s$ by $\Omega$.

Let us consider the following stabilizer
\begin{equation}
\mathbb{S}_{\Omega}=\{s\}_{s\in\Omega}.
\end{equation}
Since for all $s\in \Omega$, $s^{(Q)}$ mutually commute (and so do $s^{(\overline{Q})}$), by the virtue of \cite[Theorem 1]{Makuta2023fullynonpositive} we have that the stabilizer subspace $\mathcal{V}_{\mathbb{S}_\Omega}$ corresponding to $\mathbb{S}_{\Omega}$ cannot be entangled with respect to the bipartition $Q|\overline{Q}$. This implies that there exists $\ket{\psi_{\Omega}}\in \mathcal{V}_{\mathbb{S}_\Omega}$ such that $\ket{\psi_{\Omega}}$ is a product state with respect to the bipartition $Q|\overline{Q}$. We can then take $\ket{\phi}_{Q}$ and $\ket{\chi}_{\overline{Q}}$ such that $\ket{\phi}_{Q}\ket{\chi}_{\overline{Q}} = \ket{\psi_{\Omega}}$ which yields
\begin{equation}\label{eq:sum_s}
\begin{aligned}
\sum_{s\in\Omega}\expect{P_{s^{(Q)}}\otimes C_{s^{(Q)}}\otimes P_{s^{(\overline{Q})}}\otimes C_{s^{(\overline{Q})}}}_{\ket{\phi'}_{Q}\ket{\chi'}_{\overline{Q}}}&=\tilde{\omega}(\overline{G}_{Q}).
\end{aligned}
\end{equation}
Moreover, since $\Omega$ corresponds to the largest clique of $\overline{G}_{Q}$ it follows that for all $\tilde{s}\in \mathbb{S}\setminus \mathbb{S}_{\Omega}$ there exists $s\in \mathbb{S}_{\Omega}$ such that
\begin{equation}
[s^{(Q)},\tilde{s}^{(Q)}]_{\bullet}\neq \mathbb{1},
\end{equation}
which implies that for all $\tilde{s}\in \mathbb{S}\setminus \mathbb{S}_{\Omega}$
\begin{equation}\label{eq:sum_s'}
\expect{P_{\tilde{s^{(Q)}}}\otimes C_{\tilde{s}^{(Q)}}\otimes P_{\tilde{s}^{(\overline{Q})}}\otimes C_{\tilde{s}^{(\overline{Q})}}}_{\ket{\phi'}_{Q}\ket{\chi'}_{\overline{Q}}}=0.
\end{equation}
Substituting \eqref{eq:sum_s} and \eqref{eq:sum_s'} into Eq. \eqref{eq:sum_s_total} yields
\begin{equation}
\sum_{s\in\mathbb{S}}\expect{P_{s^{(Q)}}\otimes C_{s^{(Q)}}\otimes P_{s^{(\overline{Q})}}\otimes C_{s^{(\overline{Q})}}}_{\ket{\phi'}_{Q}\ket{\chi'}_{\overline{Q}}}=\tilde{\omega}(\overline{G}_{Q}).
\end{equation}
This allows us to conclude that
\begin{equation}
E_{\textrm{GM}}^{Q}(\mathcal{V}_{\mathbb{S}}) = 1 -d^{-k}\tilde{\omega}(\overline{G}_{Q}),
\end{equation}
which ends the proof.
\end{proof}

\section{Proof of Corollary \ref{cor:GGM}}\label{app:GGM}
\setcounter{cor}{0}
\begin{cor}
Let $\mathbb{S}=\langle g_{1},\dots, g_{k}\rangle$ be a stabilizer, such that the corresponding stabilizer subspace $\mathcal{V}_{\mathbb{S}}$ is genuinely multipartite entangled. For any such $\mathcal{V}_{\mathbb{S}}$, the generalized geometric measure of entanglement equals
\begin{equation}
E_{\textrm{GGM}}(\mathcal{V}_{\mathbb{S}})=\frac{d-1}{d}.
\end{equation}
\end{cor}
\begin{proof}
From the assumption that $\mathcal{V}_{\mathbb{S}}$ is GME it follows that for any bipartition $Q|\overline{Q}$ there exist a pair of generators $g_{i},g_{j}$ such that $[g_{i}^{(Q)},g_{j}^{(Q)}]_{\bullet}\neq\mathbb{1}$. This implies that $\rank (\gamma_{Q})>0$ for all bipartitions $Q|\overline{Q}$, since $\rank (\gamma_{Q})=0$ would imply that all $g_{i}^{(Q)}$ commute. Then, by the virtue of Lemma \ref{lem:rank_even} and Theorem \ref{thm:GM}, we have
\begin{equation}\label{eq:GM=>GMM}
\forall_{Q|\overline{Q}}\qquad E_{\textrm{GM}}^{Q}(\mathcal{V}_{\mathbb{S}}) \geqslant\frac{d-1}{d}.
\end{equation}

Let us now consider a bipartition $\{1\}|\{2,\dots,N\}$, i.e., $Q=\{1\}$. From the GME assumption and Fact \ref{fakt1} it follows that there has to exist a pair of mutually noncommuting operators $g_{i}^{(Q)},g_{j}^{(Q)}$. From the fact that every operator from the set $\{s^{(Q)}\}_{s\in\mathbb{S}}$ can be represented as a product of $g_{i}^{(Q)},g_{j}^{(Q)}$ with some scalar factor, we infer that
\begin{equation}
\rank(\gamma_{Q})=  2.
\end{equation}
Then from Theorem \ref{thm:GM} it follows that
\begin{equation}
E_{\textrm{GM}}^{Q}(\mathcal{V}_{\mathbb{S}})=\frac{d-1}{d}.
\end{equation}
and so by the virtue of Eq. \eqref{eq measure ent subspaces_2} the following holds true for all GME stabilizer subspaces $\mathcal{V}_{\mathbb{S}}$
\begin{equation}
E_{\textrm{GGM}}(\mathcal{V}_{\mathbb{S}})=\frac{d-1}{d}.
\end{equation}
\end{proof}

\section{Proof of Theorem \ref{thm:H}}\label{app:H}
\begin{thm}
Let $\mathcal{A}=\langle T_{1},\ldots T_{k}\rangle_{\odot}$ be a group as in Definition \ref{def: A} and let $d$ be an odd prime number. For each such $\mathcal{A}$ we have the following saturable upper bound
\begin{equation}
 \sum_{A\in\mathcal{A}} \langle A\rangle + \langle A^{\dagger}\rangle \leqslant 2 \tilde{\omega}(\overline{G})\left(\frac{1+\sqrt{d}}{2}\right)^{\rank(\gamma)/2}.
\end{equation}
\end{thm}
\begin{proof}
In the proof of Theorem \ref{thm:sos} we have shown that there is one-to-one correspondence between the elements of $\mathcal{S}(\mathcal{A})=\langle T_{1},\ldots,T_{q}\rangle_{\odot}$ and $\tilde{\mathbb{P}}_{q/2}$, which follows as a consequence of Theorem \ref{thm:self-testing}. The same argument allows us to write the sum over $\mathcal{A}$ as
\begin{equation}\label{eq:hamiltonian'}
 \sum_{A\in\mathcal{A}}  UA U^{\dagger} + UA^{\dagger}U^{\dagger} = \sum_{A_{C}\in \mathcal{C}(\mathcal{A})}\sum_{P\in\tilde{\mathbb{P}}_{q/2}} P\otimes C_{A_{C}}+P^{\dagger}\otimes C_{A_{C}}^{\dagger}
\end{equation}
where $U$ is the unitary given by Theorem \ref{thm:self-testing}. Using the fact that $|\langle C_{A_{C}}\rangle|\leqslant 1$ for all $A_{C}\in\mathcal{C}(\mathcal{A})$, we can compute the following upper bound
\begin{equation}\label{eq.appendixE_E3}
\sum_{A_{C}\in \mathcal{C}(\mathcal{A})}\sum_{P\in\tilde{\mathbb{P}}_{q/2}} \langle P\otimes C_{A_{C}}\rangle+\langle P^{\dagger}\otimes C_{A_{C}}^{\dagger}\rangle \leqslant \sum_{A_{C}\in \mathcal{C}(\mathcal{A})} \left(\left| \sum_{P\in\tilde{\mathbb{P}}_{q/2}} \langle P \rangle \right| + \left| \sum_{P\in\tilde{\mathbb{P}}_{q/2}} \langle P^{\dagger} \rangle \right|\right)= 2 |\mathcal{C}(\mathcal{A})|\left| \sum_{P\in\tilde{\mathbb{P}}_{q/2}} \langle P \rangle \right|.
\end{equation}
The sum over all elements of $\tilde{\mathbb{P}}_{q/2}$ is easy to compute for odd $d$ as the coefficient $\mu_{\mathbf{i},\mathbf{j}}$ equals $1$ for all $\mathbf{i},\mathbf{j}$
\begin{equation}
\sum_{P\in\tilde{\mathbb{P}}_{q/2}} P = \left(\sum_{i=0}^{d-1}X^{i}\sum_{j=0}^{d-1}Z^{j}\right)^{\otimes q/2} =\left(d^{3/2}\ket{+}\!\bra{0}\right)^{\otimes q/2},
\end{equation}
where $\ket{+}=1/\sqrt{d}\sum_{j=0}^{d-1}\ket{j}$, hence we have
\begin{equation}
\left| \sum_{P\in\tilde{\mathbb{P}}_{q/2}} \langle P \rangle \right| \leqslant \max_{\ket{\psi}\in\mathbb{C}^{d}} \left| d^{3/2} \bra{\psi}\ket{+}\!\bra{0} \ket{\psi}\right|^{q/2} 
\end{equation}
Let us consider an arbitrary state $\ket{\psi}\in\mathbb{C}^{d}$, which can be expressed in the computational basis as $\ket{\psi}=\sum_{i=0}^{d-1}a_{i}\ket{i}$, giving us
\begin{equation}
\max_{\ket{\psi}\in\mathbb{C}^{d}} \left| \bra{\psi}\ket{+}\!\bra{0} \ket{\psi}\right|=\max_{a_{0},a_{1},\ldots, a_{d-1}}\left|\sum_{i,j=0}^{d-1} \bra{i} a_{i}^{*} \ket{+}\!\bra{0} a_{j} \ket{j} \right|= \frac{1}{\sqrt{d}}\max_{a_{0},a_{1},\ldots, a_{d-1}}\left|\sum_{i=0}^{d-1}  a_{i}^{*} a_{0} \right|\leqslant \frac{1}{\sqrt{d}}\max_{a_{0},a_{1},\ldots, a_{d-1}}\sum_{i=0}^{d-1}  |a_{i}| |a_{0}|
\end{equation}
Then, using the method of Lagrange multipliers, one can show that 
\begin{equation}\label{eq:lagrange}
\frac{1}{\sqrt{d}}\max_{a_{0},a_{1},\ldots, a_{d-1}}\sum_{i=0}^{d-1}  |a_{i}| |a_{0}|= \frac{1}{2}\left(1+\frac{1}{\sqrt{d}}\right),
\end{equation}
which finally gives us
\begin{equation}\label{eq:qudit_sum_bound}
 \sum_{A\in\mathcal{A}} \langle A\rangle + \langle A^{\dagger}\rangle \leqslant 2|\mathcal{C}(\mathcal{A})|\left(\frac{d}{2}\left(1+\sqrt{d}\right)\right)^{q/2}= 2 \tilde{\omega}(\overline{G})\left(\frac{1+\sqrt{d}}{2}\right)^{\rank(\gamma)/2},
\end{equation}
where the equality follows from Lemma \ref{lem:Z(A)_cardinality}.

To show that this bound is saturable, let us first assume that there exists a state $\ket{\phi}$ such that for all $A_{C}\in \mathcal{C}(\mathcal{A})$ 
\begin{equation}\label{eq:c_assumption}
C_{A_{C}}\ket{\phi}=\ket{\phi}.
\end{equation}
Then, consider the state $\ket{\psi'}$ given by
\begin{equation}\label{eq:psi'}
\ket{\psi'} = U \left(\ket{\theta}^{\otimes q/2} \otimes \ket{\phi}\right),
\end{equation}
where the state
\begin{equation}
\ket{\theta}=\sqrt{\frac{1+\sqrt{d}}{2\sqrt{d}}}\ket{0}+ \sum_{i=1}^{d-1} \frac{1}{\sqrt{2\sqrt{d}(1+\sqrt{d})}}\ket{i},
\end{equation}
was chosen such that its coefficients satisfy Eq. \eqref{eq:lagrange}. It is easy to check that,
\begin{equation}
\sum_{P\in\tilde{\mathbb{P}}_{q/2}} \bra{\theta}^{\otimes q/2} P \ket{\theta}^{\otimes q/2} = \left(\frac{d}{2}\left(1+\sqrt{d}\right)\right)^{q/2}.
\end{equation}
Therefore, by following Eq. \eqref{eq.appendixE_E3} and Eq. \eqref{eq:c_assumption} we arrive at
\begin{equation}
\sum_{A\in\mathcal{A}}\bra{\psi'}  (A+A^{\dagger})\ket{\psi'}= 2|\mathcal{C}(\mathcal{A})|\left(\frac{d}{2}\left(1+\sqrt{d}\right)\right)^{q/2}.
\end{equation}

Consequently, if we can ensure that Eq. \eqref{eq:c_assumption} holds true, this would imply that the bound \eqref{eq:qudit_sum_bound} is saturable by $\ket{\psi'}$. To this end, let us analyze a projector $\mathcal{P}_{\mathcal{C}(\mathcal{A})}$ onto a subspace of states for which \eqref{eq:c_assumption} is satisfied. It is easy to check that this projector is given by
\begin{equation}
\mathcal{P}_{\mathcal{C}(\mathcal{A})}=\frac{1}{|\mathcal{C}(\mathcal{A})|}\sum_{A_{C}\in \mathcal{C}(\mathcal{A})} C_{A_{C}}.
\end{equation}
If we assume that Eq. \eqref{eq:c_assumption} is not satisfied for any $\ket{\varphi}$, then naturally the $\mathcal{P}_{\mathcal{C}(\mathcal{A})}$ projects onto an empty subspace, which implies that $\mathcal{P}_{\mathcal{C}(\mathcal{A})}=0$. This has important consequences, since if $\mathcal{P}_{\mathcal{C}(\mathcal{A})}=0$ then for every $A_{S}\in \mathcal{S}(\mathcal{A})$ we have
\begin{align}
\begin{split}
\sum_{A_{C}\in \mathcal{C}(\mathcal{A})}A_{S}\odot A_{C}=& A_{S} \odot U^{\dagger}\mathbb{1}\otimes\sum_{A_{C}\in \mathcal{C}(\mathcal{A})} C_{A_{C}}U\\
=& A_{S} \odot  U^{\dagger} (|\mathcal{C}(\mathcal{A})| \mathbb{1}\otimes\mathcal{P}_{\mathcal{C}(\mathcal{A})}) U=0.
\end{split}
\end{align}
We can sum the above over all $A_{S}\in\mathcal{S}(\mathcal{A})$ which yields
\begin{equation}
0=\sum_{A_{S}\in\mathcal{S}(\mathcal{A})}\sum_{A_{C}\in \mathcal{C}(\mathcal{A})}A_{S}\odot A_{C}=\sum_{A\in\mathcal{A}}A.
\end{equation}
Consequently, if $\sum_{A\in\mathcal{A}}A+A^{\dagger}\neq 0$, then $\mathcal{P}_{C(\mathcal{A})}\neq 0$ and so there exits $\ket{\phi}$ satisfying \eqref{eq:c_assumption}. Therefore, there exists the state $\ket{\psi'}$ given by Eq. \eqref{eq:psi'} that saturates the bound \eqref{eq:qudit_sum_bound}.
\end{proof}

\section{Geometric measure of entanglement of the five-qudit code across any bipartition}\label{app:five-qudit}

Here, we give a more detailed explanation of the calculation of the geometric measure of entanglement of the five-qudit code with respect to any bipartition. The five-qudit code is defined by the stabilizer $\mathbb{S}_{5}=\langle g_{1},g_{2},g_{3},g_{4}\rangle$, where
\begin{equation}
\begin{aligned}
g_{1}&= X\otimes Z\otimes Z\otimes X\otimes \mathbb{1}, \quad g_{2}=\mathbb{1}\otimes X\otimes Z\otimes Z\otimes X,\quad g_{3}=X\otimes \mathbb{1}\otimes X\otimes Z\otimes Z,\quad g_{4}=Z \otimes X \otimes \mathbb{1}\otimes X\otimes Z.
\end{aligned}
\end{equation}
Notice that the product of all of the generators equals $g_{1}g_{2}g_{3}g_{4}=Z\otimes Z\otimes X\otimes \mathbb{1}\otimes X$. This nicely showcases the symmetry of the five-qudit code under the cyclic exchange of parties, as under the permutation $1\rightarrow2\rightarrow3\rightarrow4\rightarrow5\rightarrow1$ we get the permutation of stabilizing operators $g_{1}\rightarrow g_{2}\rightarrow g_{3}\rightarrow g_{4} \rightarrow g_{1}g_{2}g_{3}g_{4}\rightarrow g_{1}$.

Now we move to computing the geometric measure of entanglement of the five-qudit code across any bipartitions $Q|\overline{Q}$ of the set of qudits. Luckily, there are actually only three geometric measures that we have to compute $Q= \{1\},\{1,2\},\{1,3\}$. This follows from two facts: the first one being that a bipartition $\{1,2\}|\{3,4,5\}$ is the same as a bipartition $\{3,4,5\}|\{1,2\}$, and so we can limit ourselves to sets $Q$ such that $|Q|\leqslant 2$. The second fact is the aforementioned symmetry of the five-qudit code, which divides the possible bipartitions into three equivalence classes 
\begin{equation}
\begin{aligned}
\{\{1\},\{2\},\{3\},\{4\},\{5\}\}&\\
\{\{1,2\},\{2,3\},\{3,4\},\{4,5\},\{1,5\}\}&\\
\{\{1,3\},\{2,4\},\{3,5\},\{1,4\},\{2,5\}\}&.
\end{aligned}
\end{equation}

Let us now write directly the truncated generators under the bipartitions $Q= \{1\},\{1,2\},\{1,3\}$
\begin{equation}
\begin{aligned}
g_{1}^{(1)}&=X,&\quad g_{2}^{(1)}&=\mathbb{1},&\quad g_{3}^{(1)}&=X,&\quad g_{4}^{(1)}&=Z,\\
g_{1}^{(1,2)}&=X\otimes Z,&\quad g_{2}^{(1,2)}&=\mathbb{1}\otimes X,&\quad g_{3}^{(1,2)}&=X\otimes\mathbb{1},&\quad g_{4}^{(1,2)}&=Z\otimes X,\\
g_{1}^{(1,3)}&=X\otimes Z,&\quad g_{2}^{(1,3)}&=\mathbb{1}\otimes Z,&\quad g_{3}^{(1,3)}&=X\otimes X,&\quad g_{4}^{(1,3)}&=Z\otimes \mathbb{1}.
\end{aligned}
\end{equation}
Then, making the assignment $T_{i}=g_{i}^{(Q)}$, we arrive at the following adjacency matrices of the generating graphs
\begin{equation}
\gamma_{1} = \begin{bmatrix}
0 & 0 & 0 & -1\\
0 & 0 & 0 & 0\\
0 & 0 & 0 & -1\\
1 & 0 & 1 & 0
\end{bmatrix},\; \gamma_{1,2} = \begin{bmatrix}
0 & 1 & 0 & 0\\
-1 & 0 & 0 & 0\\
0 & 0 & 0 & -1\\
0 & 0 & 1 & 0
\end{bmatrix},\;
\gamma_{1,3}=\begin{bmatrix}
0 & 0 & 1 & -1\\
0 & 0 & 1 & 0\\
-1 & -1 & 0 & -1\\
1 & 0 & 1 & 0
\end{bmatrix}.
\end{equation}
It is then easy to see that $\operatorname{rank}(\gamma_{1})=2$ and $\operatorname{rank}(\gamma_{1,2})=\operatorname{rank}(\gamma_{1,3})=4$. Then, by virtue of Theorem \ref{thm:GM} we have
\begin{equation}
E^{\{1\}}_{GM}(\mathcal{V}_{\mathbb{S}_{5}})= 1-d^{-1},\quad E^{\{1,2\}}_{GM}(\mathcal{V}_{\mathbb{S}_{5}})= E^{\{1,3\}}_{GM}(\mathcal{V}_{\mathbb{S}_{5}})=1-d^{-2}.
\end{equation}

\end{document}